\documentclass[aps,pra,onecolumn,notitlepage,superscriptaddress,showpacs,nofootinbib]{revtex4-1}
\usepackage{amsmath}
\usepackage{amssymb}
\usepackage{amsfonts}
\usepackage{enumerate}
\usepackage{mathrsfs}
\usepackage{graphicx}
\usepackage{epstopdf}
\usepackage{changepage}
\usepackage{bm}
\usepackage{multirow}
\usepackage{lipsum}
\usepackage{booktabs}
\usepackage{dsfont}
\usepackage{appendix}
\usepackage{physics} %\ket \bra \braket \ketbra \rank \abs \norm

\DeclareMathOperator{\wt}{\mathrm{wt}}
\DeclareMathOperator{\EE}{\mathbb{E}} %type expectation convinently
\DeclareMathOperator{\AME}{\mathrm{AME}}
\DeclareMathOperator{\Qex}{\mathrm{Qex}}
\DeclareMathOperator{\ex}{\mathrm{ex}}
\newcommand{\set}[2]{\left\{#1 \, \middle| \, #2 \right\}}
\newcommand{\card}[1]{\left| #1 \right|}
\newcommand{\F}{\mathbb{F}} %type finite field conveniently
\newcommand{\floor}[1]{\left\lfloor #1 \right\rfloor} %taking floor auto fit size
\newcommand{\CC}{\mathbb{C}}

\makeatletter  %for \autoref producing numbering with bracket
\let\oldtheequation\theequation
\renewcommand\tagform@[1]{\maketag@@@{\ignorespaces#1\unskip\@@italiccorr}}
\renewcommand\theequation{(\oldtheequation)}
\makeatother   %for \autoref producing numbering with bracket
\newtheorem{theorem}{Theorem}
\newtheorem{definition}{Definition}

\newtheorem{lemma}{Lemma}
\newtheorem{proposition}{Proposition}
\newtheorem{conjecture}{Conjecture}
\newtheorem{example}{Example}
\newtheorem{corollary}{Corollary}

\def\squareforqed{\hbox{\rlap{$\sqcap$}$\sqcup$}}
\def\qed{\ifmmode\squareforqed\else{\unskip\nobreak\hfil
		\penalty50\hskip1em\null\nobreak\hfil\squareforqed
		\parfillskip=0pt\finalhyphendemerits=0\endgraf}\fi}
\def\endenv{\ifmmode\;\else{\unskip\nobreak\hfil
		\penalty50\hskip1em\null\nobreak\hfil\;
		\parfillskip=0pt\finalhyphendemerits=0\endgraf}\fi}
\newenvironment{proof}{\noindent \textbf{{Proof.~} }}{\qed}
\def\Dbar{\leavevmode\lower.6ex\hbox to 0pt
	{\hskip-.23ex\accent"16\hss}D}
\makeatletter
\def\url@leostyle{%
	\@ifundefined{selectfont}{\def\UrlFont{\sf}}{\def\UrlFont{\small\ttfamily}}}
\makeatother
\def\bcj{\begin{conjecture}}
	\def\ecj{\end{conjecture}}
\def\bcr{\begin{corollary}}
	\def\ecr{\end{corollary}}
\def\bd{\begin{definition}}
	\def\ed{\end{definition}}
\def\bea{\begin{eqnarray}}
	\def\eea{\end{eqnarray}}
\def\bem{\begin{enumerate}}
	\def\eem{\end{enumerate}}
\def\bex{\begin{example}}
	\def\eex{\end{example}}
\def\bim{\begin{itemize}}
	\def\eim{\end{itemize}}
\def\bl{\begin{lemma}}
	\def\el{\end{lemma}}
\def\bpf{\begin{proof}}
	\def\epf{\end{proof}}
\def\bpp{\begin{proposition}}
	\def\epp{\end{proposition}}
\def\bqu{\begin{question}}
	\def\equ{\end{question}}
\def\br{\begin{remark}}
	\def\er{\end{remark}}
\def\bt{\begin{theorem}}
	\def\et{\end{theorem}}

\def\btb{\begin{tabular}}
	\def\etb{\end{tabular}}

	\newcommand{\nc}{\newcommand}
	
	\def\i{\iota}

	\nc{\bbA}{\mathbb{A}} \nc{\bbB}{\mathbb{B}} \nc{\bbC}{\mathbb{C}}
	\nc{\bbD}{\mathbb{D}} \nc{\bbE}{\mathbb{E}} \nc{\bbF}{\mathbb{F}}
	\nc{\bbG}{\mathbb{G}} \nc{\bbH}{\mathbb{H}} \nc{\bbI}{\mathbb{I}}
	\nc{\bbJ}{\mathbb{J}} \nc{\bbK}{\mathbb{K}} \nc{\bbL}{\mathbb{L}}
	\nc{\bbM}{\mathbb{M}} \nc{\bbN}{\mathbb{N}} \nc{\bbO}{\mathbb{O}}
	\nc{\bbP}{\mathbb{P}} \nc{\bbQ}{\mathbb{Q}} \nc{\bbR}{\mathbb{R}}
	\nc{\bbS}{\mathbb{S}} \nc{\bbT}{\mathbb{T}} \nc{\bbU}{\mathbb{U}}
	\nc{\bbV}{\mathbb{V}} \nc{\bbW}{\mathbb{W}} \nc{\bbX}{\mathbb{X}}
	\nc{\bbZ}{\mathbb{Z}}
	
	\nc{\bA}{{\bf A}} \nc{\bB}{{\bf B}} \nc{\bC}{{\bf C}}
	\nc{\bD}{{\bf D}} \nc{\bE}{{\bf E}} \nc{\bF}{{\bf F}}
	\nc{\bG}{{\bf G}} \nc{\bH}{{\bf H}} \nc{\bI}{{\bf I}}
	\nc{\bJ}{{\bf J}} \nc{\bK}{{\bf K}} \nc{\bL}{{\bf L}}
	\nc{\bM}{{\bf M}} \nc{\bN}{{\bf N}} \nc{\bO}{{\bf O}}
	\nc{\bP}{{\bf P}} \nc{\bQ}{{\bf Q}} \nc{\bR}{{\bf R}}
	\nc{\bS}{{\bf S}} \nc{\bT}{{\bf T}} \nc{\bU}{{\bf U}}
	\nc{\bV}{{\bf V}} \nc{\bW}{{\bf W}} \nc{\bX}{{\bf X}}
	\nc{\ba}{{\bf a}} \nc{\be}{{\bf e}} \nc{\bu}{{\bf u}}
	\nc{\brr}{{\bf r}} \nc{\bx}{{\bf x}}
	
	\nc{\cA}{{\cal A}} \nc{\cB}{{\cal B}} \nc{\cC}{{\cal C}}
	\nc{\cD}{{\cal D}} \nc{\cE}{{\cal E}} \nc{\cF}{{\cal F}}
	\nc{\cG}{{\cal G}} \nc{\cH}{{\cal H}} \nc{\cI}{{\cal I}}
	\nc{\cJ}{{\cal J}} \nc{\cK}{{\cal K}} \nc{\cL}{{\cal L}}
	\nc{\cM}{{\cal M}} \nc{\cN}{{\cal N}} \nc{\cO}{{\cal O}}
	\nc{\cP}{{\cal P}} \nc{\cQ}{{\cal Q}} \nc{\cR}{{\cal R}}
	\nc{\cS}{{\cal S}} \nc{\cT}{{\cal T}} \nc{\cU}{{\cal U}}
	\nc{\cV}{{\cal V}} \nc{\cW}{{\cal W}} \nc{\cX}{{\cal X}}
	\nc{\cZ}{{\cal Z}}
	
	\nc{\hA}{{\hat{A}}} \nc{\hB}{{\hat{B}}} \nc{\hC}{{\hat{C}}}
	\nc{\hD}{{\hat{D}}} \nc{\hE}{{\hat{E}}} \nc{\hF}{{\hat{F}}}
	\nc{\hG}{{\hat{G}}} \nc{\hH}{{\hat{H}}} \nc{\hI}{{\hat{I}}}
	\nc{\hJ}{{\hat{J}}} \nc{\hK}{{\hat{K}}} \nc{\hL}{{\hat{L}}}
	\nc{\hM}{{\hat{M}}} \nc{\hN}{{\hat{N}}} \nc{\hO}{{\hat{O}}}
	\nc{\hP}{{\hat{P}}} \nc{\hR}{{\hat{R}}} \nc{\hS}{{\hat{S}}}
	\nc{\hT}{{\hat{T}}} \nc{\hU}{{\hat{U}}} \nc{\hV}{{\hat{V}}}
	\nc{\hW}{{\hat{W}}} \nc{\hX}{{\hat{X}}} \nc{\hZ}{{\hat{Z}}}
	
	\nc{\hn}{{\hat{n}}}
	
	\def\max{\mathop{\rm max}}

	\def\rank{\mathop{\rm rank}}
	
	\def\supp{\mathop{\rm supp}}
	\usepackage[
	colorlinks,
	linkcolor = blue,
	citecolor = blue,
	urlcolor = blue]{hyperref}
	\def \qed {\hfill \vrule height7pt width 7pt depth 0pt}
	
	\newcounter{lastnote}

\begin{document}
	\title{Extremal Maximal Entanglement}
	%\date{\today}

\author{Wanchen Zhang}
% \email[]{ }
\affiliation{School of Mathematical Sciences,
	University of Science and Technology of China, Hefei, 230026,  China}
 \affiliation{Hefei National Laboratory, University of Science and Technology of China, Hefei, 230088, China}

\author{Yu Ning }
% \email[]{ }
\affiliation{Hefei National Laboratory, Hefei, 230088, China}

\author{Fei Shi}
% \email[]{}
\affiliation{Department of Computer Science, School of Computing and Data Science, University of Hong Kong, Hong Kong, 999077, China}	
	
\author{Xiande Zhang}
\email[]{Corresponding author: drzhangx@ustc.edu.cn}
\affiliation{School of Mathematical Sciences,
	University of Science and Technology of China, Hefei, 230026,  China}
 \affiliation{Hefei National Laboratory, University of Science and Technology of China, Hefei, 230088, China}

%\begin{abstract}
%Finding maximally entangled quantum states is an important topic in quantum
%entanglement research. We know that if the absolutely maximally entangled
%(AME)
%state exists, then it must be the most entangled quantum state, but what about
%the case when it doesn't? In this paper, we give the example of the maximum
%entanglement that can be maximized from a discrete point of view when the AME
%state does not exist for the case of eight qubit. Further, we give some bounds
%on the maximum number of maximally entangled parts that can be achieved for
%the
%more general case.
%\end{abstract}
\begin{abstract}
A pure multipartite quantum state is called absolutely maximally entangled
 if all reductions of no more than half of the parties are maximally mixed.
However, an $n$-qubit
absolutely maximally entangled state only exists when $n$ equals $2$, $3$, $5$, and $6$.
A natural question arises when it does not exist: which $n$-qubit pure state has the largest number of maximally mixed $\floor{n/2}$-party reductions? Denote this number by $\Qex(n)$. It was shown that $\Qex(4)=4$ in [Higuchi \emph{et al.}  \href{https://www.sciencedirect.com/science/article/pii/S0375960100004801}{Phys. Lett. A  (2000)}
%{Phys. Lett. A \textbf{273}, 213-217 (2000)}
] and $\Qex(7)=32$ in [Huber \emph{et al.} \href{https://journals.aps.org/prl/abstract/10.1103/PhysRevLett.118.200502}{Phys. Rev. Lett. (2017)}
%{Phys. Rev. Lett. \textbf{118}, 200502 (2017)}
]. In this paper, we give a general upper bound of $\Qex(n)$ by linking the
well-known Tur\'an's problem in graph theory, and provide lower bounds by constructive and
probabilistic methods.
In particular, we show that $\Qex(8)=56$, which is the third known value for this problem.
\end{abstract}
\maketitle
\vspace{-0.5cm}

\section{Introduction}\label{sec:int}
Multipartite entangled states have applications in various
quantum information tasks, such as quantum teleportation and quantum error correction
\cite{PhysRevA.87.012319,PhysRevA.69.052330}.
Therefore, the study of the entanglement properties of such states has recently become a field of intense research
%the study of the entanglement properties of such states has been a field of intense research recently
\cite{RevModPhys.80.517,Borras_2007,PhysRevA.87.012319,
	PhysRevA.77.060304,PhysRevA.91.042339,PhysRevA.86.052335,PhysRevA.92.032316,
	Brown_2005,PhysRevLett.118.200502}.
For a pure quantum state with multiple parties, the maximal entanglement exists
between a bipartition if the reduction to the smaller part is
maximally mixed.  Multipartite  states that exhibit maximal
entanglement across all possible bipartitions are known as
\emph{absolutely maximally entangled} (AME) states \cite{PhysRevA.86.052335}.
For a fixed number of parties,
AME states always exist when the local dimension is large enough \cite{Feng2017}.
%\footnote{需不需要强调local dimension是prime power? 我没有仔细看引用的文章。
	%	如果是用MDS码构造，需要local dimension是prime power?}
However, when the number of parties is large enough, AME states do not
exist for a fixed local dimension
\cite{PhysRevA.69.052330}.
%\footnote{这句可以改成：However, when the local dimension is fixed,
	%	AME states never exist as long as the number of parties is large enough.}
Taking qubit states as an example,
AME states do not exist when the number of
parties is $4$ or greater than $6$
\cite{HIGUCHI2000213,PhysRevA.69.052330,PhysRevLett.118.200502,rains1998quantum,rains1999quantum,nebe2006self}.
So a natural question arises: When an AME state does not exist,
which state can take the place of an AME state in the  quantum information
tasks above? %{\color{red}(is it possible to consider an application of your states at the end?)}

As AME states share the full number of bipartitions where maximal entanglement lives, a pure state with the largest possible number of maximally mixed half-body reductions would be a good candidate. Such states have been studied over the past two decades when one met the nonexistence of AME qubit states. Higuchi and Sudbery \cite{HIGUCHI2000213} demonstrated that, for a $4$-qubit state, at most four out of ${\binom{4}{2} =6}$ two-party reductions can be maximally mixed, provided that all one-party reductions are
maximally mixed. For the $7$-qubit case, Huber \emph{et al.} \cite{PhysRevLett.118.200502} showed that up to 32 three-party reductions can be maximally mixed, given that all two-party reductions are maximally mixed.

Pure states with the largest  number of maximally mixed half-body reductions are worth studying for another reason. There are many ways to measure multipartite entanglement \cite{PhysRevA.61.052306,PhysRevA.63.044301,PhysRevA.74.022314,PhysRevA.69.052330}, where different measures are often inconsistent because they employ
different strategies, focus on different aspects, and capture different
features of this quantum phenomenon. However, despite different entanglement measures,  the possible maximal entanglement for any bipartition is achieved only when the reduction is maximally mixed. When AME states do not exist, one common way is to look for pure states that maximize the average entanglement among all bipartitions \cite{Zha_2012,HIGUCHI2000213,Brown_2005,Borras_2007,Zha_2011}. Another way is to look for pure states that achieve maximal entanglement between as many bipartitons as possible, that is,  pure states that are close to AME states from a discrete point of view \cite{HIGUCHI2000213,PhysRevLett.118.200502}.

In this paper, we study the largest number of maximally mixed half-body reductions in an arbitrary qubit pure state.  We connect this number with the well-known Tur\'an's number in graph theory and
establish an upper bound on the largest number of maximally mixed
half-body reductions that one pure state can have. Based on this upper bound, we show that in a pure state of eight qubits, at most $56$ many four-party reductions can be maximally mixed. Such a state can be constructed from orthogonal arrays  and graphs \cite{PhysRevA.99.042332,sudevan2022n}. This is another nontrivial extremal case besides the $4$-qubit and $7$-qubit states.  General lower bounds are also given by explicitly constructing graph states or in a probabilistic way. Comparing with the existing results under the average linear entropy, it is interesting to find that the quantum state with the largest number of maximally mixed half-body reductions also has the largest average linear entropy among $4$, $7$, and $8$-qubit pure states \cite{PhysRevA.69.052330,PhysRevA.77.060304}.

The rest of this paper is organized as follows. In \autoref{sec:pre}, the preliminary knowledge is introduced. In \autoref{sec:upperbound}, we give an upper bound on the number of maximal mixed reductions through
an extremal problem in combinatorics.
 In \autoref{sec:graphstate}, we give a general lower
bound by constructing good graph states and a lower bound by the probabilistic method.
In particular, the maximum number of maximally mixed half-body reductions
that an $8$-qubit pure state may possess is determined, and
examples of $8$-qubit pure states reaching this upper bound are constructed.
In \autoref{sec:com}, we examine some examples of pure states having
the largest number of maximally mixed $\floor{n/2}$-party reductions.
These examples also tend to share the largest average linear entropy.
Finally, we conclude in \autoref{sec:con}.

\section{Preliminaries and Problem statement}\label{sec:pre}
Let $[n]:= \{ 1,\ldots, n \}$ and let $\binom{[n]}{k}$ denote subsets of size $k$ of $[n]$.
\subsection{Problem statement}
First, we give the definition of $k$-uniform states.
\begin{definition}
	A pure state $\ket{\psi} \in \qty(\CC^d)^{\otimes n}$  shared among $n$ parties
	in $[n]$ is said to be \emph{$k$-uniform}, where $k \le \floor{n/2}$ is
	a positive integer, if the reductions of $\ket{\psi}$ to any
$m$-party with $m \le k$ are maximally mixed, i.e.,
all reductions of $\ket{\psi}$ of size $m$ are the
same, namely $\frac{(I_d)^{\otimes m}}{d^m}$.
%{\color{red}(this identity $\II$ is not consistent with  $ I_2$ in section B, and also not consistent with $I_k$ in Eq (18). )}

\end{definition}
%\red{可以加一句，$k$有界$k \le \floor{n/2}$，达到这个界时候，叫AME。这样引入
%更自然一点}
%For a pure state $\ket{\psi} \in \pqty{\CC^d}^{\otimes n}$ and
%$k \in [1,n]$,

The existence of $k$-uniform states has been widely studied \cite{Feng2017,PhysRevA.94.012346,PhysRevA.99.042332,Zang_2021,PhysRevA.104.032601,shi2023boundskuniformquantumstates},
while their existence is ensured when $n$ is large for a fixed $k$ and $d$ \cite{Feng2017}.
%(right?)cite{Feng and Jin's work}?
However, this is not the case when $k$ is related to $n$ \cite{shi2023boundskuniformquantumstates}.
Specifically, when $k=\lfloor\frac{n}{2}\rfloor$, an $\lfloor\frac{n}{2}\rfloor$-uniform
state in $(\bbC^{d})^{\otimes n}$, also known as an AME state and denoted by $\AME(n,d)$,
is very rare for the given local dimension $d$.
%\red{AME(n,d)的记号应该在数学环境中，
	%	在导言区定义了}\verb|\DeclareMathOperator{\AME}{\mathrm{AME}}|
%\red{后面就用}\verb|$\AME(n,d)$|
%We denote such a state as an AME $(n,d)$ state.
%AME states are very rare for given local dimensions.
%In this paper, we only
%consider qubit states, that is local dimension $d=2$.
%\footnote{可以不用这么急着说只考虑d=2，下面也介绍了d=4的结果，可以多介绍一点}
In the case of $d=2$, $\AME(n,2)$ exists only for $n = 2,3,5$, and $6$
\cite{rains1999quantum,PhysRevLett.118.200502,PhysRevA.69.052330}. $\AME(4,2)$ was proved  not to exist in \cite{HIGUCHI2000213},
	and $\AME(n,2)$ for $n \ge 8$  was proved  not to exist in
	\cite{rains1998quantum,rains1999quantum,nebe2006self,PhysRevA.69.052330}.
	The last
	case  $\AME(7,2)$ was proved not to exist by Huber
	\emph{et al.} in \cite{PhysRevLett.118.200502}, where the authors provided
	a method for characterizing qubit AME states and their approximations,
	making use of the Bloch representation \cite{RevModPhys.29.74}. The method in \cite{PhysRevLett.118.200502} will be recalled in the next subsection and will be applied later in our proofs.

Now, we introduce the terminologies for our problem on qubit states,  which can be easily generalized to qudit states.

For a pure state $\ket{\psi} \in \pqty{\CC^2}^{\otimes n}$  and
$k \in [n]$,
we denote by $\mathcal{M}_k(\ket{\psi})$ the set of $k$-party to which
the reductions of $\rho$ are maximally mixed,
where $\rho = \ketbra{\psi}{\psi}$.
Namely,
\[\mathcal{M}_k(\ket{\psi}) \triangleq \set{A \in \binom{[n]}{k} }{\rho_A = \frac{(I_2)^{\otimes k}}{2^k}, \rho = \ketbra{\psi}{\psi}}.\]
%$$
%\mathcal{M}_k(\rho) \triangleq \set{\rho_A \subset \mathcal{D}_k(\rho) }{\rho_A = \frac{\II}{\sqrt{d^k}}}.
%$$
Denote  $m_k(\ket{\psi})$ the size of $\mathcal{M}_k(\ket{\psi})$.
%For convenience,  let $M_k(\ket{\psi}) \triangleq M_k(\rho)$ and $m_k(\ket{\psi}) \triangleq m_k(\rho)$.
We define the \emph{quantum extremal number}, denoted by $\Qex(n,k)$,
to be the maximum  $m_k(\ket{\psi})$ among all pure states $\psi \in
\pqty{\CC^2}^{\otimes n}$, i.e.,
$$
\Qex(n,k) \triangleq \max_{\ket{\psi} \in \pqty{\CC^2}^{\otimes n}}
m_k(\ket{\psi}).
$$
When $k=\floor{n/2}$, we write $\Qex(n)$ for short. By \cite{PhysRevLett.118.200502,HIGUCHI2000213}, we know $\Qex(4)=4$ and $\Qex(7)=32$.

There is a good reason for the terminology ``quantum extremal number'',
as $\Qex(n,k)$ will be proved later to be related to Tur\'an's extremal number in graph theory.
If $\ket{\psi} \in
\pqty{\CC^2}^{\otimes n}$ is a pure state reaching
the quantum extremal number, i.e., $m_k(\ket{\psi}) = \Qex(n,k)$ for some
$k \in \big[\floor{n/2}\big]$, then $\ket{\psi}$ is said to be an
\emph{$k$-extremal maximally entangled ($k$-EME) state}.
Note the fact that a pure state is $(k+1)$-EME does not mean that it is $k$-EME.
Furthermore, if
$\ket{\psi}$ is $k$-EME for all $k \in \big[\floor{n/2}\big]$, then
we call $\ket{\psi}$ a \emph{perfect extremal maximally entangled (PEME) state}.
Clearly, AME states are PEME states, since the values of $m_k$ achieve the trivial upper bound $\binom{n}{k}$ for all $k \in \big[\floor{n/2}\big]$.
Define \[\pi(n,k) \triangleq \frac{\Qex(n,k)}{\binom{n}{k}}\] as the density of maximally mixed $k$-party reductions.
If $k$ is a  fixed integer, then $\lim_{n \to \infty} \pi(n,k) = 1$
\cite{Feng2017}. However, when $k$ is a function of $n$, for example
$k = \Theta(n)$, this is no longer the case. The behavior of $\pi(n,k)$
will be studied in this paper for $k = n/2$ as $n$ goes to
infinity.

%we use their methods to give upper bounds on general $ex(n,d,k)$, some of which are restatements of their results, and some of which are obtained by further careful analysis. Here we will briefly describe this approach and give some important lemmas to be used subsequently.
%We will derive some results in the form of the Bloch representation of a quantum state, similar to that done in \cite{PhysRevLett.118.200502}.
\subsection{The Bloch representation of $k$-uniform states}
In this subsection we briefly introduce the Bloch representation of quantum states and the parity rule lemma given in \cite{PhysRevLett.118.200502}.

Any $n$-qubit state can be written in terms of tensor products of Pauli matrices as
\begin{equation}\label{eqrho}
\rho=\sum_{\alpha_1, \ldots, \alpha_n}\frac{1}{2^n}r_{\alpha_1, \ldots,
\alpha_n}\sigma_{\alpha_1}\otimes\ldots\otimes\sigma_{\alpha_n}
\end{equation}
with
\begin{equation}\label{rrr}
r_{\alpha_1, \ldots,
\alpha_n}=\tr(\sigma_{\alpha_1}\otimes\ldots\otimes
\sigma_{\alpha_n}\times\rho),
\end{equation}
where
$\alpha_i \in \{0, x, y, z\}$, $\sigma_0=I_2=$ $\begin{pmatrix} 1 & 0 \\ 0 & 1 \end{pmatrix}$, $\sigma_x=$ $\begin{pmatrix} 0 & 1 \\ 1 & 0 \end{pmatrix}$,  $\sigma_y=$ $\begin{pmatrix} 0 & -i \\ i & 0 \end{pmatrix}$ and $\sigma_z=$ $\begin{pmatrix} 1 & 0 \\ 0 & -1 \end{pmatrix}$.
 For convenience, denote
$\sigma_{\alpha}:=\sigma_{\alpha_1}\otimes\ldots\otimes\sigma_{\alpha_n}$ with
$\alpha=(\alpha_1,\ldots,\alpha_n)\in \{0, x, y, z\}^n$.
%Often we will write $X_j , Y_j , Z_j$ for the Pauli matrices acting on
%particle $j$ alone.
Define the \emph{support} of $\sigma_{\alpha}$ as
$\supp(\sigma_{\alpha})=\{i\mid \alpha_i\neq 0 \ \text{for} \ 1\leq i\leq n\}$,
and the weight of $\sigma_{\alpha}$ as $\wt(\sigma_{\alpha}) =
|\supp(\sigma_{\alpha})|$.
Let $P_j$ denote the sum of the terms $\sigma_{\alpha}$ with wt($\sigma_{\alpha})
= j $ in \autoref{eqrho}. Consequently, the state can be expressed as
\begin{equation}
	\rho=\frac{1}{2^n}(I_2^{\otimes n}+ \sum_{j=1}^{n} P_j).
\end{equation}
To be more specific, we denote by $P^{(\mathcal{J})}_j$ the partial sum in $P_j$ whose support is $\mathcal{J}\subset [n]$. For example, a state of four qubits
reads
\begin{equation}\label{eq5}
	\rho=\frac{1}{2^4}(I_2^{\otimes4}+\sum_{i=1}^{4}{P}^{(i)}_1 + \sum_{1\le j<k\le 4}P^{(jk)}_2 + \sum_{1\le l<p<q\le 4}P^{(lpq)}_3 + P_4),
\end{equation}
where, e.g., $P_2^{(12)}=\sum r_{\alpha_1, \alpha_2, 0, 0} \sigma_{\alpha_1}\otimes\sigma_{\alpha_2}\otimes I_2\otimes I_2$ and $\alpha_1, \alpha_2 \neq 0$.
%A density matrix $\rho$ on n qubits can be expanded in terms of the Pauli basis as
%\begin{equation}\label{pauliexpand}
%	\rho=2^{-n}\sum_{\sigma_{\alpha}\in\mathcal{P}_n} Tr[\sigma_{\alpha}\rho]\sigma_{\alpha}.
%\end{equation}
%Given a state $\rho$, we obtain its reduction on parties $S$ by acting with the partial trace on its complement $S^c$
%\begin{equation}
%	\rho_S=\text{Tr}_{S^c}(\rho).
%\end{equation}
%
%Then we have the fact that $\text{Tr}_{S^c}(\sigma_{\alpha})=0$ if supp($\sigma_{\alpha}$)$\nsubseteq S$.
%For example,
% When tracing out the fourth qubit in Eq.~(\ref{eq5}), one drops the terms $P_4$, $P^{(mn4)}_3$, $P^{(k4)}_2$, and $P^{(4)}_1$, as they do not contain an identity in the fourth subsystem.
% Also, the normalization prefactor is multiplied by the dimension of the parties over which the partial trace was performed, resulting in
% \begin{equation}
% 	\text{tr}_{\{4\}}[\rho]\otimes  I_2=\frac{1}{2^3}( I_2^{\otimes3}+\sum_{j=1}^{4}{P}^{(j)}_1 + \sum_{1\le k<l\le 4}P^{(kl)}_2 + \sum_{1\le m<n<o\le 4}P^{(mno)}_3 + P_4).
% \end{equation}

When $\rho = \ketbra{\psi}{\psi}$ is a $k$-uniform state, the coefficients in the Bloch representation of terms with
weight $1, 2, \ldots , k$ are zero
\cite{PhysRevA.87.012319,sudevan2022n}, that is, $P_1=P_2=\dots=P_{k}=0$.
Another important property of $\rho$ follows from the Schmidt decomposition: the complementary reductions of any bipartition share the same
spectrum. Since a reduction to $\mathcal{J}\subset [n]$ of $\rho$ with  $\card{\mathcal{J}}=l\le k$  is maximally mixed, its
complementary reduction $\rho_{\bar{\mathcal{J}}}$ with $\bar{\mathcal{J}}: =[n]\setminus \mathcal{J}$
% {\color{red}(the notation $\mathcal{J}^\complement$ is not good-looking, is it better to replace it by $\bar{\mathcal{J}}$?{\color{blue}this is corrected})}
 of size
$n-l \ge \lfloor\frac{n}{2}\rfloor$
%\footnote{$k$在刚刚用来指代$k$-uniform, 这里可以换个记号}
has all $2^{l}$
nonzero eigenvalues equal to $\lambda=2^{-l}$.
%Thus the reduction is
%proportional to a projector,
As analyzed in \cite{PhysRevLett.118.200502}, we have
\begin{equation}\label{eqrhoA2}
	\rho_{\bar{\mathcal{J}}}^2=2^{-l}\rho_{\bar{\mathcal{J}}}
\end{equation}
%By the Schmidt decomposition, one further sees that the full state $\ket{\psi}=\ket{\psi}_{\bar{\mathcal{J}}\mathcal{J}} $ is an eigenvector of the reduction $\rho_{\bar{\mathcal{J}}}$,
and
\begin{equation}\label{eqschdec}
	\rho_{\bar{\mathcal{J}}}\otimes I_2^{\otimes l}\ket{\psi}=2^{-l}\ket{\psi}.
	%\rho_{\bar{\mathcal{J}}}\otimes I_2^{\otimes m}\ket{\psi}_{\bar{\mathcal{J}}\mathcal{J}}=2^{-m}\ket{\psi}_{\bar{\mathcal{J}}\mathcal{J}}.
\end{equation}
%\red{上面这段介绍Huber的方法，需要用自己的话重新叙述一遍，现在感觉和
%	Huber的原文读起来差不多。}

%Accordingly, for a $k$-uniform state having all reductions to any $k$ parties are maximally mixed, any reductions to $n-l$ parties $\rho_{\bar{\mathcal{J}}}$ with
%$n-k  \le \card{\bar{\mathcal{J}}} \le n $ fulfills relations \autoref{eqrhoA2} and
%\autoref{eqschdec}. Note that both relations hold for general dimension $d$.

Finally, we restate the parity rule lemma from \cite{PhysRevLett.118.200502}, which will play a key role when recognizing what terms $P_i$ may appear in $\rho^2$ in the Bloch representation.
%Further discussion depends on recognizing what terms may arise in $\rho^2$.
%The following observation plays a key role.

\begin{lemma}[parity rule \cite{PhysRevLett.118.200502}]\label{rule}
	Let $M$, $N$ be Hermitian %\footnote{不需要强调Hermitian吧}
  operators
	proportional to $n$-fold tensor
	products of single-qubit Pauli operators, $M =
	c_M\sigma_{\alpha_{\mu_1}}\otimes\cdots\otimes\sigma_{\alpha_{\mu_n}}$,
	$N=c_N\sigma_{\alpha_{\nu_1}}\otimes\cdots\otimes\sigma_{\alpha_{\nu_n}}$,
	where $c_M, c_N \in \mathbb{R}$. Then, if the anticommutator
	$\{M, N\} := MN + NM$ of $M$ and $N$
	does not vanish, its weight fulfills
%	{\color{red}(this weight notation $|\cdot|$ is not consistent with $wt(\cdot)$)}
	\begin{equation}
	\wt(\acomm{M}{N}) \equiv \wt(M) + \wt(N) \pmod{2}.
	\end{equation}
%\red{这个公式应该这样打：}
%\verb|\wt\acomm{M}{N} \equiv \wt(M) + \wt(N) \pmod{2}|,
%\red{结果是：$\wt\pqty{\acomm{M}{N}} \equiv \wt(M) + \wt(N) \pmod{2}$}
\end{lemma}

\subsection{Hypergraphs and Tur\'an's extremal number}
Now we introduce related concepts in hypergraphs.
A \emph{hypergraph} ${H}$ is a pair
$(V,E)$, where $V$ is a set of elements called \emph{vertices},  and
$E$ is a set of subsets of $V$ called \emph{hyperedges}. If
every hyperedge in $E$ has the same size $k$, then
${H}$ is called \emph{$k$-uniform}. A
$2$-uniform hypergraph is just a simple graph.
Let ${H}_1=(V_1,E_1)$ be another $k$-uniform hypergraph.
If $V \subset V_1$ and $E \subset E_1$, we say ${H}$ is a \emph{sub-hypergraph} of
${H}_1$.  If ${H}_1$ contains no copy of ${H}$ as a sub-hypergraph,
we say ${H}_1$ is ${H}$-\emph{free}.

The Tur\'an's  extremal number concerns how many edges
${H}_1$ may possess, provided that ${H}_1$
is  ${H}$-free.
More precisely: let $2 \le k \le n$ be fixed positive integers and
${H}$ be a fixed $k$-uniform hypergraph, the \emph{Tur\'an's extremal number}
is defined as
$$
\ex_k\pqty{n,{H}} \triangleq \max
\set{ |E_1|}{{H}_1 = (V_1,E_1) \text{ is }
	n\text{-vertex, } k\text{-uniform and } {H}\text{-free}}.
$$

The $k$-uniform $l$-vertex complete hypergraph, denoted as $K_l^k$,
is the hypergraph with vertex set $[l]$ and edge set $\binom{[l]}{k}$.
The extremal number $\ex_k\pqty{n,K_l^k}$
for $K_{l}^k$ can be interpreted in another way.
For $k \le l \le n$, define $T(n,l,k)$ to be the smallest number of $k$-subsets of an $n$-set $X$, such that every $l$-subset of $X$ contains at least
one of the $k$-subsets. Considering these $T(n,l,k)$ many  $k$-subsets as hyperedges,  we can see that $\ex_k\pqty{n,K_l^k} = \binom{n}{k} - T(n,l,k)$.
Moreover, we have the following bound on $T(n,l,k)$.
\begin{proposition}[\cite{keevash2011hypergraph,de1983extension}]
	\label{turannumber222}
	For all positive integers $k \le l \le n$,
	$T (n, l, k) \ge \frac{n-l+1}{n-k+1}\binom{n}{k}/\binom{l-1}{k-1}$.
\end{proposition}
%{\color{blue}$T (n, l, k) \ge \frac{n-l+1}{n-k+1}\binom{n}{k}/\binom{l-1}{k-1}=\frac{n!(l-k)!(n-l+1)}{(n-k+1)!(l-1)!k}=\frac{n-l+1}{k}\binom{n}{k-1}/\binom{l-1}{k-1}$}

By \autoref{turannumber222}, $\ex_k\pqty{n,K_l^k}\leq \binom{n}{k} -\frac{n-l+1}{n-k+1}\binom{n}{k}/\binom{l-1}{k-1}$.

\section{Connections of quantum and Tur\'{a}n's extremal numbers} \label{sec:upperbound}
%\red{标题：删去construct。
%	Turan number --> Turan's extremal number或者extremal number.
%	因为是叫quantum extremal number，所以经典的那个最好也叫什么extremal number。
%	总的来说，标题可以是: Correspondence between quantum extremal
%	number and extremal number
%}
%In this section, we establish the connections between quantum extremal
	%number and Tur\'{a}n's extremal number, and provide upper bounds of the quantum extremal
	%number.

Huber \emph{et al.} \cite{PhysRevLett.118.200502} showed that if $\ket{\psi}$
is a
pure state of seven qubits with all $2$-reductions of $\ketbra{\psi}$
maximally mixed, then the number of maximally mixed $3$-reductions of
$\ketbra{\psi}$ is at most $32$. In this section, by generalizing the
results of \cite{PhysRevLett.118.200502}, we
provide  upper bounds for the quantum extermal number in terms of
Tur\'an's extremal number.

First, we associate each
pure state  with a uniform hypergraph. Let $\ket{\psi} \in (\CC^2)^{\otimes n}$ be a
pure state of $n$ qubits, and $\rho = \ketbra{\psi}$. For $k \in [n]$,
we defined a $k$-uniform hypergraph $G_k(\ket{\psi})$ as follows:
the vertex set is $[n]$, i.e., each party of
$\ket{\psi}$ corresponds to a vertex of $G_k(\ket{\psi})$;
for any $k$-subset $\mathcal{A} \subset [n]$, $\mathcal{A}$ is
an edge of $G_k(\ket{\psi})$ if and only if $\rho_\mathcal{A}$ is
maximally mixed. Under these notations, several results in \cite{PhysRevLett.118.200502} can reformulated as follows.

\begin{lemma}[%	{\color{red}indicate the corresponding lemma or place in Huber's paper, same below},
Cases 1 and 2 of Appendix B in \cite{PhysRevLett.118.200502}]\label{nonevenmod42}
Let $\ket{\psi}$ be an $n$-qubit pure state, where $n=2k$, $k \geq 2$ and $k\neq 3$.
% Let $\mathcal{I} =
%\setnd{1,2,\ldots, 4n+2}$ with $n \ge 2$.
For any $\mathcal{A} \subset [n]$ with
$\card{\mathcal{A}} = k+2$, there exists $\mathcal{B} \subset \mathcal{A}$
with $\card{\mathcal{B}} = k$ such
	that the reduction of $\ket{\psi}$ to $\mathcal{B}$ is not maximally mixed.

	%\red{可以叙述更严谨、清晰一点：
		%Let $\ket{\psi}$ be a pure state of $4n+2$ qubits with $n \ge 2$ and
		%$\rho=\ketbra{\psi}{\psi}$. For any $A \subset [4n+2]$ with
		%$\card{A} = 2n+3$, there exists $B \subset A$ with $\card{B} = 2n+1$ such
		%that $\rho_B$ is not maximally mixed.
		%}
	%\red{一开始我觉得应该叙述更严谨、清晰一点。但是如果要这么修改，好像后面整个证明都要改，
		%还挺长的，所以我先跳过一部分证明。看你怎么决定了；你要是想严谨一点，可以先把
		%证明对应改一遍，我再读读。要是就保持现在这样的状态，我回头再看现在的证明。}
\end{lemma}

 The idea for
	the proof of \autoref{nonevenmod42} originates from
	\cite{PhysRevLett.118.200502}. For completeness, we include a
	detailed proof of \autoref{nonevenmod42} in Appendix \ref{appendix1}. By \autoref{nonevenmod42},   if
	$\ket{\psi}$ is a pure state of $2k$ qubits, then
	$G_{k}(\ket{\psi})$ is $K_{k+2}^{k}$-free.
	Thus,
	\begin{equation}\label{eq4m2}
		\Qex (2k) \le \ex_{k} (2k,K_{k+2}^{k}).
	\end{equation}
	
%By \autoref{nonevenmod42},   if
%$\ket{\psi}$ is a pure state of $4m+2$ qubits, then
%$G_{2m+1}(\ket{\psi})$ is $K_{2m+3}^{2m+1}$-free.
%Thus,
%\begin{equation}\label{eq4m2}
%  \Qex (4m+2) \le \ex_{2m+1} (4m+2,K_{2m+3}^{2m+1}).
%\end{equation} The idea for
%the proof of \autoref{nonevenmod42} originates from
%\cite{PhysRevLett.118.200502}. For completeness, in {\color{red}\autoref{appendix1}(should be appendix)} we include a
%detailed proof.

%Similarly, we also have the following result.
%\begin{theorem}[\cite{PhysRevLett.118.200502}]\label{nonevenmod400}
%Let $\ket{\psi}$ be an $n$-qubit pure state, where $n=4m$ and $m\geq 2$.
%%Let $\ket{\psi}$ be a pure state of $n=4m$ qubits with $m \ge 2$.
%%Let
%%$\mathcal{I} = \setnd{1,2,\ldots, 4n}$ with $n \ge 2$.
%For any $\mathcal{A} \subset [n]$ with $\card{\mathcal{A}} = 2m+2$,
%there exists $\mathcal{B} \subset \mathcal{A}$ with $\card{\mathcal{B}} = 2m$
%such that the reduction of $\ket{\psi}$ to $\mathcal{B}$ is not maximally mixed.
%%Consider a pure state of $4n$ qubits. There cannot exist a reduction
%%$\rho_{(2n+2)}$ to $2n+2$ parties such that every reduction to $2n$ parties
%%obtained from the $\rho_{(2n+2)}$ is maximally mixed for any $n\ge2$.
%\end{theorem}

%Then, \autoref{nonevenmod400} equivalently states that if $\ket{\psi}$ is a
%pure state of $4m$ qubits, then $G_{2m}(\ket{\psi})$ is $K_{2m+2}^{2m}$-free.
%Thus,
%\begin{equation}\label{eq4m}
%  \Qex(4m) \le \ex_{2m}(4m,K_{2m+2}^{2m}).
%\end{equation}
Combining \autoref{eq4m2}  and \autoref{turannumber222}, we have the following result for any even $n=2k$.

\begin{corollary} For any $k \geq 2$ and $k\neq 3$,
\begin{equation}\label{eq2k}
  \Qex(2k) \le \ex_{k} (2k,K_{k+2}^{k}) \le \binom{2k}{k} -
	\frac{k-1}{k+1}\binom{k+1}{k-1}^{-1}\binom{2k}{k}.
\end{equation}

\end{corollary}

When $k=2$, \autoref{eq2k} gives $\Qex(4)\le 5$. However, this is not tight since $\Qex(4)=4$. When $k=4,5,6$, \autoref{eq2k} gives $\Qex(8)\le 65$, $\Qex(10)\le 240$, and $\Qex(12)\le 892$.

%
%\begin{align*}
%	\Qex(4m,2m) &\le \ex_{2m} (4m,K_{2m+2}^{2m}) \le \binom{4m}{2m} -
%	\frac{2m-1}{2m+1}\binom{2m+1}{2m-1}^{-1}\binom{4m}{2m}, \\
%	\Qex(4m+2,2m+1) &\le \ex_{2m+1} (4m+2,K_{2m+3}^{2m+1}) \le \binom{4m+2}{2m+1} -
%	\frac{2m}{2m+2}\binom{2m+2}{2m}^{-1}\binom{4m+2}{2m+1}.
%\end{align*}

For the case where the number $n$ of parties is odd, the subgraph-free property of the corresponding hypergraph is a bit complicated. We combine the cases $n=4m+1$ and $4m+3$ in \cite{PhysRevLett.118.200502} into \autoref{thmodd}.
%the hypergraph which $H(\ket{\psi})$ is free of is a bit complicated.
 For completeness, we include a
	detailed proof in Appendix \ref{appendix2}.
\begin{lemma}[Cases 3 and 4 of Appendix B in \cite{PhysRevLett.118.200502}]\label{thmodd}
  Let $\ket{\psi}$ be an $n$-qubit pure state with $n=2k+1$, $k \ge 3$ and $k\neq 5$.
For any $\mathcal{A} \subset [n]$ with
$\card{\mathcal{A}} = k+2$, there exists a $k$-subset
$\mathcal{B} \subset [n]$ with $|\mathcal{B}\cap \mathcal{A}|=1$ or $k$ such
that the reduction  of $\ket{\psi}$ to $\mathcal{B}$ is not maximally mixed.
\end{lemma}

In \autoref{thmodd}, the subset $\mathcal{B}$ is either contained in $\mathcal{A}$ or contains exactly one-party of $\mathcal{A}$. So for any odd $n=2k+1$, we can define a $k$-uniform hypergraph $H_k$ as follows: the vertex set is $[n]$, and for a fixed subset $\mathcal{A}\subset [n]$ of size $k+2$,
 a $k$-subset $\mathcal{B}\subset [n]$ is a hyperedge if $|\mathcal{B}\cap \mathcal{A}|=1$ or $k$.
% \begin{itemize}
%   \item $\mathcal{B}$ is any $k$-subset  of $\mathcal{A}$; or
%   \item $\mathcal{B}$ contains exactly one vertex from $\mathcal{A}$, and all other vertices are outside $\mathcal{A}$.
% \end{itemize}
 When $k=2$, $H_2$ is just the simple complete graph on five vertices. Then, for odd $n=2k+1$, $G_{k}(\ket{\psi})$ is $H_k$-free for any $n$-qubit pure state $\ket{\psi}$, and hence $\Qex(2k+1) \le \ex_{k}(2k+1,H_k)$. Next, we give a simple upper bound of this Tur\'an's  extremal number.

\begin{proposition}\label{turannumberodd}
	For any $k\geq 2$,
	$\ex_k(2k+1, H_k) \le \binom{2k+1}{k} - \left\lceil\binom{2k+1}{k+2}\big/  \left(\binom{k+1}{2} + k\right) \right\rceil$.
%
%	$ex(n, G_H) \le \binom{n}{\lfloor\frac{n}{2}\rfloor} - \lceil\binom{n}{\lfloor\frac{n}{2}\rfloor+2}/  (\binom{\lfloor\frac{n}{2}\rfloor+1}{\lfloor\frac{n}{2}\rfloor-1} + \lfloor\frac{n}{2}\rfloor) \rceil$.
\end{proposition}
\begin{proof} We prove the upper bound by starting from the complete $k$-uniform hypergraph $K_{2k+1}^{k}$ and counting how many hyperedges have to be removed to make the resultant  $H_k$-free. We need to count how many copies of $H_k$ are corrupted after removing a hyperedge.

Suppose we remove a hyperedge $\mathcal{B}$, which is a subset of size $k$. Then we count the number of copies of $H_k$ that containing $\mathcal{B}$ as a hyperedge. There are two cases. If $|\mathcal{B}\cap \mathcal{A}|=k$, then there are at most $\binom{k+1}{2}$ such $H_k$'s. Otherwise, if $|\mathcal{B}\cap \mathcal{A}|=1$, then there are at most $k$ such $H_k$'s. So removing a  hyperedge $\mathcal{B}$ will corrupt at most $\binom{k+1}{2}+k$ copies of $H_k$. Since there are $\binom{2k+1}{k+2}$ copies of $H_k$ in $K_{2k+1}^{k}$,  at least $\left\lceil\binom{2k+1}{k+2}\big/  (\binom{k+1}{2} + k) \right\rceil$ hyperedges have to be removed from $K_{2k+1}^{k}$ to make it  $H_k$-free. So $\ex_k(2k+1, H_k) \le \binom{2k+1}{k} - \left\lceil\binom{2k+1}{k+2}\big/  \left(\binom{k+1}{2} + k\right) \right\rceil$.

\end{proof}

 By \autoref{turannumberodd}, we have for $k \ge 3$ and $k\neq 5$,
 \begin{equation}\label{eq2k1}
   \Qex(2k+1) \le \ex_{k}(2k+1,H_k) \le \binom{2k+1}{k} - \left \lceil\binom{2k+1}{k+2} \bigg/  \left(\binom{k+1}{2} + k\right) \right \rceil.
 \end{equation}

%
%\begin{equation}
%	Qex (4n+1,2n) \le ex_{2n}(4n+1, K_{2n+2}^{2n}, G_H) \le \binom{4n+1}{2n} - \lceil\binom{4n+1}{2n+2}/  (\binom{2n+1}{2n-1} + 2n) \rceil.
%\end{equation}
%
%\begin{equation}
%	Qex (4n+3,2n+1) \le ex_{2n+1}(4n+3, K_{2n+3}^{2n+1}, G_H) \le  \binom{4n+3}{2n+1} - \lceil\binom{4n+3}{2n+3}/  (\binom{2n+2}{2n} + 2n+1) \rceil.
%\end{equation}

When $k=3$, \autoref{eq2k1} gives  $\Qex(7) \le \ex_3 (7, H_3) \le
\binom{7}{3}-\left\lceil\binom{7}{5}\big/  \left(\binom{4}{2} + 3\right) \right\rceil = 32$, which is tight since $\Qex (7) = 32$
\cite{PhysRevLett.118.200502}. For $k=4$,  Zha \emph{et al.} \cite{Zha_2020} showed that there exist a $9$-qubit pure state with $110$ maximally mixed $4$-body reductions, then we have $110 \le \Qex(9) \le \ex_3 (9, H_4) \le 120$.

Finally, we mention that we don't assume that the pure state is $(\lfloor n/2\rfloor-1)$-uniform when deducing these bounds in this section, while this was assumed in  \cite{Zha_2018,Zha_2020}.
%all bounds stated in this section are based on weaker assumptions, that is, we don't assume that the pure state is $(\lfloor n/2\rfloor-1)$-uniform when deducing these bounds,  which was assumed  in
 However, when a state achieves certain upper bound, it must be $(\lfloor n/2\rfloor-1)$-uniform in some cases. See \autoref{2meme} in \autoref{sec:graphstate}.

\section{Improving the upper bound of $4m$-qubit}\label{sec:PEME}
%In the previous section we gave general upper bounds for $\Qex(n)$,
%but these bounds tend to be hard to reach, and in this section we will improve on
%some special cases to make these bounds easier to reach.
%
%The upper bound for the $n=4m$ case can be improved even further after a
%more refined analysis than the proof of \autoref{nonevenmod42} in \cite{PhysRevLett.118.200502},
%as follows.
In the previous section, we gave general upper bounds for $\Qex(n)$ by applying results in  \cite{PhysRevLett.118.200502}.
In this section, we  improve the upper bound for the case $n=4m$ by a more refined analysis than the proof of
\autoref{nonevenmod42}.
%In this section, we focus on the case when the number $n$ of parties is a multiple of $4$. Write $n=4m$ with $m\geq 2$. First, we improve the upper bound of \autoref{eq2k} for this case. Then we construct an eight-qubit state which achieves the new upper bound, that is a $4$-EME state with eight qubits. Finally, we show that any pure state achieving this upper bound must be a $(2m-1)$-uniform state. Hence, a $4$-EME state with eight qubits must be a PEME.

%\subsection{Improving the upper bound of $4m$-qubit}

%We use a more refined analysis than the proof of \autoref{nonevenmod42} in \cite{PhysRevLett.118.200502} to  obtain the following result.
\begin{theorem}\label{nonevenmod4}
Let $\ket{\psi}$ be an $n$-qubit pure state, where $n=4m$ and $m\geq 2$.
For any $\mathcal{A} \subset [n]$ with
$\card{\mathcal{A}} = 2m+1$, there exists $\mathcal{B} \subset \mathcal{A}$ with $\card{\mathcal{B}} = 2m$ such
that the reduction of $\ket{\psi}$ to $\mathcal{B}$ is not maximally mixed.
%	Consider a pure state of $4n$ qubits. There cannot exist a reduction $\rho_{(2n+1)}$ to $2n+1$ parties such that every reduction to  $2n$ parties obtained from the $\rho_{(2n+1)}$ is maximally mixed.
\end{theorem}
%\red{这里对比前面的结论有改进？  如果有改进，为什么不直接叙述这个结论，而不再
%叙述前面的结论？}
%\red{这里有类似Theorem 1的问题。 你要不要把这个定理的陈述、证明写的更严谨一点。
%你可以先决定，你决定了我再看。}
\begin{proof}
The proof is by contradiction.	Without loss of generality, assume $\mathcal{A} = [2m+1]$
and the reduction of $\ket{\psi}$ to any $2m$ parties in $\mathcal{A}$ is maximally mixed.
Notice the fact that for a system of $4m$ parties, if the reduction to $\mathcal{J}$ is
maximally mixed and $\card{\mathcal{J}}=2m$, then the reduction to $\bar{\mathcal{J}}$ is also maximally mixed.
Therefore, from the fact that the reduction to $ [2m]$ is maximally mixed, we get that the
reduction to $[4m]\setminus[2m]$, and hence the reduction to $\bar{\mathcal{A}}$, is maximally mixed.
	
Thus the reduction to $\bar{\mathcal{A}}$ has all ${2^{2m-1}}$ nonzero eigenvalues equal to $\lambda={2^{1-2m}}$.
By
\autoref{eqrhoA2}, the reduction to $\mathcal{A}$ is proportional to a projector,
%{\color{blue}(make the format of referring an equation consistent throughout.
%This is corrected.). }
\begin{equation}\label{rho2n+1pro}
	\rho_{\mathcal{A}}^2={2^{1-2m}}\rho_{\mathcal{A}}.
\end{equation}
Since all reductions to $2m$-party obtained from $\mathcal{A}$ are maximally mixed,
we can expand the reduction to $\mathcal{A}$ in the Bloch representation,
\begin{equation}\label{rhoBB2n+1}
	\rho_{\mathcal{A}}=\frac{1}{2^{2m+1}}( I_2+P_{2m+1}).
\end{equation}
Combining \autoref{rho2n+1pro}  and \autoref{rhoBB2n+1}, we obtain
\begin{equation}\label{eqp2n+1}
	( I_2+P_{2m+1})( I_2+P_{2m+1})=4( I_2+P_{2m+1}).
\end{equation}

	By applying the parity rule outlined in \autoref{rule}, we observe that only specific products on the left-hand side of \autoref{eqp2n+1} can contribute to \( P_{2m+1} \) on the right-hand side. Notably, the term \( P_{2m+1}^2 \) on the left-hand side does not contribute to \( P_{2m+1} \) on the right-hand side, as dictated by \autoref{rule}. Consequently, we can gather all terms of odd weight from both sides of \autoref{eqp2n+1} to derive:
%	Now we apply the parity rule in \autoref{rule}: Only certain products occurring on the left-hand
%	side of \autoref{eqp2n+1} can contribute to $P_{2m+1}$ on the right-hand side. Indeed, $P_{2m+1}^2$
%	on the left-hand side cannot contribute to $P_{2m+1}$ on the right-hand side by \autoref{rule}.
%	Thus we can collect terms of odd weight on both sides of \autoref{eqp2n+1} and get
	\begin{equation}
		2P_{2m+1}=4P_{2m+1}.
	\end{equation}
	So $P_{2m+1}=0$, which means $\rho_{\mathcal{A}}=\frac{1}{2^{2m+1}} I_2$,
	a contradiction.
\end{proof}

Similar to the analysis in the previous section, here we get the following result for $m\geq 2$,
\begin{equation}\label{eqq}
	\Qex (4m) \le \ex_{2m}(4m,K_{2m+1}^{2m}) \le \binom{4m}{2m} - \frac{1}{2m+1}\binom{4m}{2m}=\binom{4m}{2m-1}.
\end{equation}
The bound in \autoref{eqq} improves the one in \autoref{eq4m2} also in a combinatorial way: forbidding a smaller hypergraph $K_{2m+1}^{2m}$ leads to less edges in the hypergraph. Note that the latter equality in \autoref{eqq} holds if and only if the
$2m$-uniform hypergraph $G_{2m}(\ket{\psi})$ satisfies the following property:
let $\bar{E}$ be the set of $2m$-subsets that are not hyperedges,
then $|\mathcal{A}\cap \mathcal{B}|\leq 2m-2$ for any $\mathcal{A}\neq \mathcal{B}\in \bar{E}$,
or equivalently, each $\mathcal{A}\in \bar{E}$ corrupted a different $K_{2m+1}^{2m}$ \cite{de1983extension}.

By \autoref{eqq}, we have $\Qex (8) \le \ex_4(8, K_{5}^{4}) \le 56$  and $\Qex (12) \le 792$, which greatly improve those from \autoref{eq2k}.

 In the next section, we will show that some quantum states are $4$-EME states by using the improved bound.

\section{A Construction by graph state}\label{sec:graphstate}

In this section, we construct an $8$-qubit state which achieves \autoref{eqq}, that is 
a $4$-EME state with eight qubits. Then, we show that any $4m$-qubit pure state achieving this upper bound must be a $(2m-1)$-uniform state. Hence, a $4$-EME state with eight qubits must be a PEME state.
Furthermore, we construct families of graph states and estimate their values $m_k$. These estimations provide lower bounds for $\Qex(n)$.
 Finally, we give a lower bound of $\Qex(n,k)$ from random graph states by a probabilistic method.

%
%\subsection{Preliminary on graph states}
%%\red{这一个Subsection要不要放到Preliminary}
First, we introduce the definition of a graph state formalized by adjacency matrices \cite{helwig2013absolutely}. Let $G$ be a simple graph with vertex set $[n]$ and $A=(a_{ij})_{n\times n} $ be its adjacency matrix. Let $\widetilde{A} = (\widetilde{a_{ij}})_{n\times n}$ with $\widetilde{a_{ij}} = a_{ij}$ for $i \le j$ and $0$ otherwise. Then the corresponding graph state is defined as \[\ket{G}:=\frac{1}{\sqrt{2^n}}\sum_{c \in \mathbb{Z}_2^n }(-1)^{c\widetilde{A}c^T} \ket{c},\] where $c$ is a row vector.

By Theorem~$4$ of Ref.~\cite{Feng2017}, we can easily decide which $k$-reduction of $\ket{G}$ is maximally mixed.

\begin{corollary}[\cite{Feng2017,helwig2013absolutely}]
	\label{fengkre}
Let $G$, $A$ and $\ket{G}$ be defined as above. For any $K\subset[n]$ of size $ k\le n/2$,
if the $k\times (n-k)$ submatrix of $A$ with rows in $K$ and columns in $\bar{K}$,
denoted by $A_{K\times \bar{K}}$, has rank $k$ over $\mathbb{F}_2$, then the reduction of
$\ket{G}$ to $K$ is maximally mixed.
%	Let $1 \le k \le n/2$. If there is an adjacency matrix $A$ of simple graph such that for
%	subset $K=\{i_1, i_2, \ldots, i_k\}$ of $\{1, 2, \ldots, n\}$ with $|K| = k$, the submatrix $A_{K\times \bar{K}}$ has rank $k$ over $Z_2$ , then the reduction of $n$-qudit graph state $\sum_{c \in \mathbb{Z}_2^n }\varphi(c) \ket{c}$ corresponding to $A$ to parties $\{i_1, i_2, \ldots, i_k\}$ is maximally mixed.
\end{corollary}

\subsection{A construction of a PEME state in $(\mathbb{C}^2)^{\otimes 8}$}\label{sub:4eme}

By \autoref{fengkre}, to estimate $m_k(\ket{G})$, it is enough to count the number of $k$-subsets $K\subset [n]$ such that $A_{K\times \bar{K}}$ has a full rank. For example, the graph $T_4$ with eight vertices in \autoref{fig1} has the following adjacency matrix,
\begin{equation}
	A=\left[\begin{array}{llllllll}
		0 & 1 & 1 & 1 & 1 & 0 & 0 & 0  \\
		1 & 0 & 1 & 1 & 0 & 1 & 0 & 0  \\
		1 & 1 & 0 & 1 & 0 & 0 & 1 & 0  \\
		1 & 1 & 1 & 0 & 0 & 0 & 0 & 1  \\
		1 & 0 & 0 & 0 & 0 & 1 & 1 & 1  \\
		0 & 1 & 0 & 0 & 1 & 0 & 1 & 1  \\
		0 & 0 & 1 & 0 & 1 & 1 & 0 & 1  \\
		0 & 0 & 0 & 1 & 1 & 1 & 1 & 0  \\
	\end{array}\right].
\end{equation}
It can be checked that there are $56$ subsets $K$ of four rows such that $A_{K\times \bar{K}}$ has  rank four.
Then the graph state $\ket{T_4}$ is a $4$-EME in $(\mathbb{C}^2)^{\otimes 8}$, since $m_4(\ket{T_4})=56$, achieving the upper bound in \autoref{eqq}.

\begin{figure}
	\centering
	\includegraphics[width=0.3 \textwidth]{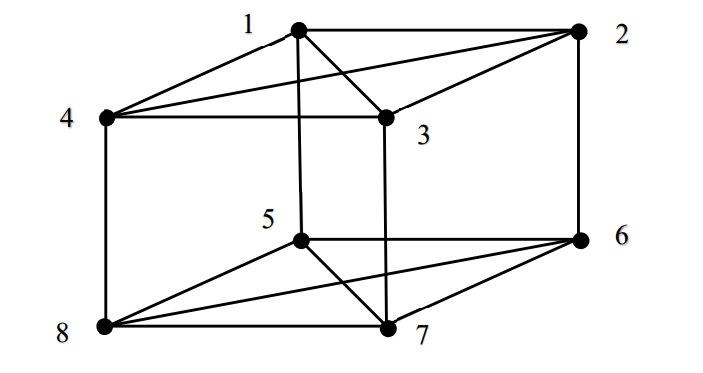}
	\caption{a PEME state of eight qubits.}\label{fig1}
\end{figure}

%\subsection{Properties of states reaching the upper bound}
Next, we show that $\ket{T_4}$ is $3$-uniform.
Here, we study a more general problem: for integers $s<k$, how large does $m_k(\ket{\psi})$ need to be to ensure that the pure state $\ket{\psi}$ is $s$-uniform?
% and here, we study a general problem, for  integers $s<k$, how large does $m_k(\ket{\psi})$ have to be to ensure that the pure state $\ket{\psi}$ is $s$-uniform?
 Indeed, we have the following observation: if $m_k(\ket{\psi})>\binom{n}{k}-\binom{n-s}{k-s}$, then $\ket{\psi}$ must be $s$-uniform.
This is because if an $s$-party reduction is not maximally mixed, then any reduction to a $k$-party containing this  $s$-party is not maximally mixed; % a maximally mixed $k$-party reduction
%implies that all $m$-party reductions among these $k$ parties
%are maximally mixed;
when $m_k(\ket{\psi})>\binom{n}{k}-\binom{n-s}{k-s}$, these maximally mixed $k$-party reductions must cover all
possible $s$-party reductions.
Further, when $n = 2k$, we have  that  $\ket{\psi}$ is $s$-uniform
if  $m_k(\ket{\psi})>\binom{n}{k}-2\binom{n-s}{k-s}$; this is
because maximally mixed $k$-party reductions always occur in pairs in this case.

Next, we consider the special case when $n=4m$ and $s=2m-1$. We give an example first.
\begin{theorem}\label{4EME}
	Any $4$-EME state  in $(\mathbb{C}^2)^{\otimes 8}$ is $3$-uniform, and thus a PEME state.
\end{theorem}
\begin{proof} By \autoref{sub:4eme}, any $4$-EME state $\ket{\psi}$ in $(\mathbb{C}^2)^{\otimes 8}$ has $m_4(\ket{\psi})=56$, which achieves the upper bound in \autoref{eqq}. Note that the upper bound of \autoref{eqq} is achieved if and only if every two $4$-reductions $\mathcal{A}$ and $\mathcal{B}$ that are not  maximally mixed satisfy $|\mathcal{A}\cap \mathcal{B}|<3$ by the remark after \autoref{eqq}.
	However if $\ket{\psi}$ is not $3$-uniform, say the reduction to the three-party $\{1,2,3\}$ is not maximally mixed,
	then the reductions to
	four-party $\{1,2,3,4\}$, $\{1,2,3,5\}$, $\{1,2,3,6\}$, $\{1,2,3,7\}$ and
	$\{1,2,3,8\}$ are not maximally mixed. A contradiction to $|\mathcal{A}\cap \mathcal{B}|<3$.
\end{proof}

Thus, we are able to prove that the graph state $\ket{T_4}$ is a PEME state. Note that the construction of PEME states in $(\mathbb{C}^2)^{\otimes 8}$
is not unique, the $3$-uniform quantum state constructed from orthogonal arrays by Li \emph{et al.} \cite{PhysRevA.99.042332} is also a PEME state. Later we will show that the PEME states constructed in these two ways are not LU-equivalent.
Indeed \autoref{4EME} can be generalized to any pure state  in $(\mathbb{C}^2)^{\otimes 4m}$ achieving \autoref{eqq}, whose proof is similar and thus omitted.
\begin{theorem}\label{2meme}Let  $m\geq 2$.
	Any pure state  in $(\mathbb{C}^2)^{\otimes 4m}$ achieving \autoref{eqq}  is $(2m-1)$-uniform, and thus a PEME state.
\end{theorem}

By \autoref{2meme}, if there does not exist a $(2m-1)$-uniform state in $(\mathbb{C}^2)^{\otimes 4m}$ for some $m$, then we can infer that the upper bound in \autoref{eqq} cannot be reached. There are  many works on the existence of $k$-uniform quantum states.
For example,   Rains' bound \cite{rains1999quantum} says that  a $k$-uniform quantum state in $(\mathbb{C}^2)^{\otimes 6j+l}$ exists only if $k\leq 2j+1 $ when $ 0\leq l<5$ or $k \leq 2j+2$ when $l=5$. By \autoref{2meme}, this means the  upper bound  in \autoref{eqq} is not tight when $m\geq 4$.

However, for an aritrary graph $G$, it is usually difficult to determine $m_k(\ket{G})$. Next, we construct a family of special graphs, whose corresponding graph states can be analysed.

\subsection{A lower bound of $\Qex(2k)$ from explicit graph states}\label{sub2k}

Let $n=2k$ with $k\geq 2$. We generalize the graph in \autoref{fig1} to  a graph $T_k$ as follows. The vertex set consists of two parts $B=\{b_1, b_2, \ldots, b_k\}$ and $C=\{c_1, c_2, \ldots, c_k\}$. The graph $T_k$ induces a complete graph of size $k$ on both $B$ and $C$, and a matching $\{b_i,c_i\},i\in[k]$ between $B$ and $C$. Then its adjacency matrix has the following form,

%More generally, we can use this construction to get a lower bound for $Qex(2n,n)$. The graph corresponding to this quantum state is a graph with $2n$ vertices which is divided into two sets of vertice sets $B$ and $C$, where $|B| = |C| = n$. The induced subgraphs of both set $B$ and set $C$ are complete graphs, and the set of edges between set $B$ and set $C$ constitutes exactly one matching on the graph. Then its corresponding adjacency matrix has the following form:

\begin{equation}
A=	\left[\begin{array}{ll}
		J_k - I_k & \quad I_k  \\
	\quad	I_k & J_k - I_k  \\
	\end{array}\right],
\end{equation}
where $J$ is the matrix with all elements $1$, and $I$ is the identity matrix. Rows and columns of $A$ are indexed by $\{b_1, b_2, \ldots, b_k,c_1, c_2, \ldots, c_k\}$.

Now we consider $K$ as a $k$-subset of $B\cup C$ in different cases. If $K=B$ or $C$, then  $A_{K\times \bar{K}}=I_k$, which has full rank.

Assume $K\neq B$ and $K\neq C$. For convenience,  let $B_S:=\{b_s: s\in S\}$ for any subset $S\subset [k]$. Similarly we define $C_S$.
Then $K$ can be written as $K=B_S\cup C_{S'}$ with $S,S'$ proper subsets of $[k]$ satisfying $|S|=s$ and $|S'|=k-s$. By symmetry,  we assume $1 \leq k/2 \leq s\leq k-1$. Then  $A_{K\times \bar{K}}$ has the following form,
\begin{equation}\label{akk}
	A_{K\times \bar{K}}=	\left[\begin{array}{ll}
		 \,J_{s\times (k-s)}& P_1  \\
		P_2 & \,J_{(k-s)\times s} \\		
	\end{array}\right],
\end{equation}
%{\color{red}(Here, the notation $H_i$ is conflict to the one in Eq (11), which is a global notation. so better change the one here.)}
where $P_1$ is an $s\times s$ matrix,
and $P_2$ is a $(k-s)\times (k-s)$ matrix.
Let $l=|S\cap \bar{S'}|$ and $l'=|S'\cap \bar{S}|$.
Suppose $|S\cap S'|=i\leq k-s$, then $l=s-i$ and $l'=k-s-i$.
Then there are exactly $l$ ones in $P_1$ distributed in distinct rows and
columns. Similarly, $P_2$ has exactly $l'$ ones distributed in
distinct rows and columns.
% Since either $s$ equals $0$ or $k$ corresponds to the trivial case, which implies that the case $k = 1$ is trivial, in the following we only need to consider the nontrivial case.
 To make it easier to understand, we will first present the cases $k = 2$ and $3$, and then generalize to the case where $k$ is greater than or equal to $4$.

Case $1$: $k = 2$. Then $s$ must be $1$. If $i=1$, then $S=S'$, and $P_1=P_2=0$ is a
$1 \times 1$ matrix. Thus the rank of $A_{K\times \bar{K}}$ over $\F_2$ is $2$. We have two choices of such $K$. If $i=0$, then $S=\bar{S'}$, and $P_1=P_2=1$ is a $1 \times 1$ matrix. Thus $A_{K\times \bar{K}}$ is not of full rank.
  Combining the cases $K=B$ or $C$, we have $m_2(\ket {T_2}) = 4$,  which is equal to $Qex(4)$.

Case $2$: $k = 3$. Then $s$ must be $2$. If $i=1$, then $P_2=0$ is a
$1 \times 1$ matrix and $P_1$ has exactly one entry with $1$ and all others with $0$. It can be checked that the rank of $A_{K\times \bar{K}}$ over $\F_2$ is $3$. If $i=0$, then $P_1$ is a permutation matrix of rank $2$, and $P_2 = 1$ is a $1 \times 1$ matrix. It can be checked that the rank of $A_{K\times \bar{K}}$ over $\F_2$ is $3$ as well.  Thus $m_2(\ket {T_3}) =2+2\times (\binom{3}{2}\times 2\times(3-2)+\binom{3}{2})= \binom{6}{3}=20$, which implies that $\ket {T_3}$ is an AME state.

Case $3$: $k \ge 4$.
%We only need to consider $s \ge k/2$ by symmetry.
If $|S\cap S'|=i\geq 2$. Then $l\le s-2$, that is, $P_1$ has at least two zero columns. Thus $A_{K \times \bar{K}}$ has two identical columns and is not of full rank.
%In this case, $A_{K \times \bar{K}}$ is not of full rank.
%Because $P_1$ has at least two zero columns, and these two
%columns in $A_{K \times \bar{K}}$ are identical.

If $i=1$, then  $l = s-1$ and $l'=k-s-1$.
By swapping rows or columns appropriately, ${A}_{K\times \bar{K}}$ can always be written in the form in  \autoref{akk} with
\begin{equation}
	{P}_1=	\left[\begin{array}{ll}
		I_{(s-1)\times(s-1)}&\, 0  \\
		0 & \,0 \\			
	\end{array}\right] \ \text{and}\
	{P}_2=	\left[\begin{array}{ll}
		I_{(k-s+1)\times(k-s+1)}&\, 0  \\
		0 & \,0 \\			
	\end{array}\right].
\end{equation}
We can then easily check that ${A}_{K\times \bar{K}}$ is of full rank.
It can be calculated that for each fixed $s$, the number of $K$'s that satisfy $i=1$ is $s(k-s)\binom{k}{s}$. In total, there are  $2\times \sum_{s=\lceil k/2 \rceil}^{k-1}s(k-s)\binom{k}{s}=\sum_{s=1}^{k-1}s (k-s)\binom{k}{s}$ many $K$'s.

%\begin{equation}
%	\widetilde{A}_{K\times \bar{K}}=	\left[\begin{array}{ll}
	%			 \,J_{s\times (k-s)}& \widetilde{H}_1  \\
	%			\widetilde{P}_2 & \,J_{(k-s)\times s} \\			
	%		\end{array}\right],
%\end{equation}
%where
%\begin{equation}
%	\widetilde{H}_1=	\left[\begin{array}{ll}
	%		I_{\{s-1\}\times\{s-1\}}&\, 0  \\
	%		0 & \,0 \\			
	%	\end{array}\right] ,
%		\widetilde{H}_2=	\left[\begin{array}{ll}
	%		I_{\{k-s+1\}\times\{k-s+1\}}&\, 0  \\
	%		0 & \,0 \\			
	%	\end{array}\right].
%\end{equation}

If $i=0$, then $l = s$. In this case, $P_1$ is a permutation matrix of rank $s$ and $P_2$ is a permutation matrix of rank $k-s$.
By swapping rows or columns appropriately, we can make $P_1$ and $P_2$ identical matrices. %${A}_{K\times \bar{K}}$ can always be written  in the form in  \autoref{akk} with
%
%
%\begin{equation}
%	\widetilde{A}_{K\times \bar{K}}=	\left[\begin{array}{ll}
	%		\,J_{s\times (k-s)}& I_{\{s\}\times\{s\}}  \\
	%		I_{\{k-s\}\times\{k-s\}} & \,J_{(k-s)\times s} \\			
	%	\end{array}\right],
%\end{equation}
We apply the following elementary transformations to determine the
rank of ${A}_{K \times \bar{K}}$.
First, by adding the last $k-s$ rows to each of the $i$-th row, for $i \in [s]$,
we change ${A}_{K \times \bar{K}}$ to
\begin{equation}
	\begin{bmatrix}
		0_{s \times (k-s)} & I_{s \times s} + (k-s)J_{s \times s} \\
		I_{(k-s) \times (k-s)} & J_{(k-s) \times s}
	\end{bmatrix}.
\end{equation}
Clearly the bottom-right part $J_{(k-s) \times s}$ can be
cancelled out by the first $k-s$ columns. So
\begin{equation}
	\rank {A}_{K \times \bar{K}} =
	k-s + \rank\qty(I_{s \times s} + (k-s)J_{s \times s}).
\end{equation}
It is easy to check that only when $k-s$ is even, or when  $k-s$ is odd and $s$ is even, $\rank\qty(I_{s \times s} + (k-s)J_{s \times s})=s$.
%If $k-s$ is even, we have
%\begin{equation}
%	\rank {A}_{K \times \bar{K}} =
%	k-s + \rank\qty(I_{s \times s}) = k.
%\end{equation}
%If $k-s$ is odd, then
%\begin{equation}
%	\rank {A}_{K \times \bar{K}} =
%	k-s + \rank\qty(I_{s \times s} + J_{s \times s}).
%\end{equation}
%We know that $I_{s \times s} + J_{s \times s}$ is of full rank
%if and only if $s$ is even.
So  in this case, ${A}_{K \times \bar{K}}$ is
of full rank except when $k$ is even and $s$ is odd.
For each fixed $s$, the number of $K$'s that satisfy $i=0$ is $\binom{k}{s}$. In total, if $k$ is odd, then the number of such $K$'s with a full rank ${A}_{K\times \bar{K}}$ when $i=0$ is $\sum_{s=1}^{k-1}\binom{k}{s}$; if $k$ is even, then the number of $K$'s with a full rank ${A}_{K\times \bar{K}}$ when $i=0$ is $\sum_{s=2, s \text{ is even } }^{k-2}\binom{k}{s}$.

\begin{figure}
	\centering
	\includegraphics[width=0.3 \textwidth]{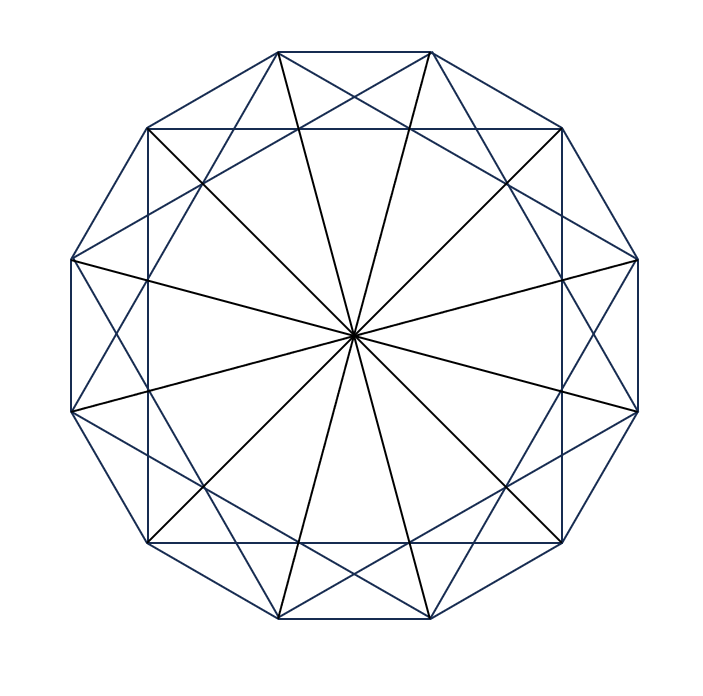}
	\caption{A $(1,3,6)$-circulant graph.}\label{fig2}
\end{figure}

\vspace{0.3cm}

Combining all pieces, we conclude that
%{\color{red} (double check the following formulas. input the number in each of the above case, so that we can check the following total number. I don't when $k=2,3$, the proof can give this bound. You need to check it.)}
\begin{equation}\label{eqlbg}
 \Qex(2k)\geq m_k(\ket {T_k}) =
\begin{cases}
\sum_{s=1}^{k-1}\binom{k}{s}(s(k-s)) + \sum_{s=2, s \text{ is even } }^{k-2}\binom{k}{s}+2
=2^{k-2}(k^2-k+2) &  \text{ $k$ even}, \\
\sum_{i=0}^{k}\binom{k}{s}(1+s(k-s))=2^{k-2}(k^2-k+4) & \text{$k$ odd}.
\end{cases}
\end{equation}

% can compute that this particular graph with $2k$
%vertices has a total number
%$\sum_{i=1}^{k-1}\binom{k}{i}(i(k-i)) + \sum_{i \text{ mod } 2 = 0}
%\binom{k}{i}$ of $k$-party bipartition that are maximally mixed when
%$k$ is even,
%and a total number $\sum_{i=0}^{k}\binom{k}{i}(1+i(k-i))$ of
%$k$-party bipartions that are maximally mixed when $k$ is odd.

\autoref{eqlbg}
 gives a nice lower bound for $\Qex(2k)$ when $k$ is small.
%	{\color{red} , (use examples to say how good it is when $k$ is small)}
	For example, when $k \leq 4$, $\Qex(2k) = m_k(\ket {T_k})$. In addition,  $\Qex(10) \ge m_5(\ket {T_5} = 192$, while the upper bound is $240$, and $\Qex(12) \ge m_6(\ket {T_6} = 512$, while the upper bound is $792$. For $k = 6$, we can construct a $5$-uniform graph state which has $540$ maximally mixed $6$-body reductions, which is the best lower bound we can find. See the corresponding graph in \autoref{fig2}. Danielsen 	\emph{et al.} showed that this graph state is the only $5$-uniform quantum state among $12$-qubit stabilizer states \cite{DANIELSEN20061351}. Combined with \autoref{2meme}, this also shows that the upper bound $792$ for $12$-qubit states is not reachable among stabilizer states.
	
	However, $m_k(\ket {T_k})$ as a lower bound becomes worse as $k$ tends to infinity, as we can note that
%	 $\lim_{k \to \infty}\pi(2k,k)\ge$
	$\lim_{k \to \infty}2^{k-2}(k^2-k+2)\binom{2k}{k}^{-1}=0$ and
%	$\lim_{k \to \infty}\pi(2k,k)\ge$
	$\lim_{k \to \infty}2^{k-2}(k^2-k+4)\binom{2k}{k}^{-1}=0.$ In the next subsection we will use probabilistic methods to give a relatively good lower bound when $k$ tends to infinity.

\subsection{A lower bound of $\Qex(n,k)$ from random graph states}

In this subsection, we consider the random graph $G=G(n,1/2)$, that is a random graph with $n$ vertices such that any pair of vertices are adjacent with probability $1/2$. Let $\ket{G}$ denote the graph state associated with $G$, and $A$ denote the adjacency matrix of $G$, and
$\rho = \ketbra{G}{G}$. Now we compute the expected number of maximally mixed $k$-party reductions of $\rho$.

%In the following, we will approach a lower bound for $Qex(n,k)$ with a probabilistc method. In particular, we consider the expectataion value of
%the number of maximally mixed $k$-party reductions of graph states associated with random graphs, and this expectataion value will eventually
%provide with a lower bound for $Qex(n,k)$. The introduction of this probabilistc method to study combinatorial problems was a huge success
%and we refer to \cite{Alon2008ThePM} for the basics and uncountable applications of this probabilistc method.

%Consider $G = G(n,1/2)$, that is an $n$-vertex graph where any pair of vertices are adjacent independently with probability $1/2$.
%%Clearly, $G$ coincides with the uniform distribution over the space of all $n$-vertex graphs.
%Let $\ket{G}$ denote the graph state associated with $G$, and $A$ denote the adjacency matrix of $G$, and
%$\rho = \ketbra{G}{G}$.
%More precisely, $A$ is some probability distribution determined by $G$ over the space of zero-diagonal symmetric matrices in
%$\F_2^{n \times n}$,  while $\ket{G}$ is the same probability distribution as $A$ except that the space of $\ket{G}$
%is the space of graph states associated with each
For a subset $K \subset [n]$ of $k$ parties with $k \le \floor{\frac{n}{2}}$, it is easy to see that $A_{K \times \bar{K}}$  satisfies the uniform distribution
over $\F_2^{k \times (n-k)}$, since the edges of $G$ are chosen independently and uniformly with probability $1/2$.
For positive integers $r \le s$, let $f(r,s)$ denote the number of matrices of rank $r$ in $\F_2^{r \times s}$.
It is known that $f(r,s) = \prod_{l = 0}^{r-1}\qty(2^s-2^l)$.
By \autoref{fengkre}, $\rho_K$ is maximally mixed if and only if $\rank A_{K \times \bar{K}} = k$.
So,
$$
\Pr\qty[\rho_K \text{ is maximally mixed}] = \Pr\qty[\rank A_{K \times \bar{K}} = k]
= \frac{f(k,n-k)}{2^{k(n-k)}} =\prod_{l=0}^{k-1}\qty(1-2^{l-n+k}).
$$
Let $X_K$ be the indicator random variable for the event that $\rho_K$ is maximally mixed, i.e.,
$$
X_K =
\begin{cases}
1 & \rho_K \text{ is maximally mixed}, \\
0 & \rho_K \text{ is not maximally mixed}.
\end{cases}
$$
%Clearly, the expectation of $X_K$ is
%$$
%\EE(X_K) = \Pr\lrb{\rho_K \text{ is maximally mixed}} = \Pr\lrb{\rank A_{K \times K^\complement} = k} = \frac{f(k,n-k)}{2^{k(n-k)}}
%$$
Let $X = \sum_{K \in {[n] \choose k}}X_K$. By linearity of expectataion, we have
$$
\EE(X) = \sum_{K \in {[n] \choose k}}\EE(X_K) =
\sum_{K \in {[n] \choose k}}\Pr\qty[\rho_K \text{ is maximally mixed}]={n \choose k}
\prod_{l=0}^{k-1}\qty(1-2^{l-n+k}).
$$
So there exists an $n$-vertex graph, whose associated graph state has
at least ${n \choose k}\prod_{l=0}^{k-1}\qty(1-2^{l-n+k})$ many
 $k$-party reductions that are maximally mixed.  In other words, we have shown that
$$
\Qex(n,k) \ge {n \choose k}\prod_{l=0}^{k-1}\qty(1-2^{l-n+k}).
$$
It thus follows that
%\footnote{We should define $\pi(n,k) = Qex(n,k)/{n \choose k}$ in the preliminary, and discuss
%values like $\lim_{n \to \infty}\pi(n,n/2)$. Here, $\pi(n,k)$ follows this definition.}
$$
\pi(n,k) \geq \prod_{l=0}^{k-1}\qty(1-2^{l-n+k}).
$$
Denote
$L(n,k) = \prod_{l=0}^{k-1}\qty(1-2^{l-n+k})$.
If $n = 2k$, we have $\lim_{k \to \infty}L(2k,k)
= \prod_{l=1}^{\infty}\qty(1-\frac{1}{2^l}) \simeq 0.288788095$
\cite{finch2003mathematical,OEIS048651}. So, we conclude that
$$
\lim_{k \to \infty}\pi(2k,k) \geq 0.288788095.
$$
\section{comparison under the potential of multipartite entanglement}\label{sec:com}
In this section, we relate our problem of investigating $\Qex(n)$ to a well-studied problem in the literature.

For $n$-qubit states, a \emph{maximally multipartite entangled state} (MMES) is defined to be a minimizer of the potential of multipartite entanglement by Facchi \emph{et al.} \cite{PhysRevA.77.060304}:
%Consider an $n$-qubit pure state, Facchi \emph{et al.}  \cite{PhysRevA.77.060304} defined a maximally multipartite entangled state (MMES), which is a minimizer of what they call the potential of multipartite entanglement,
 \begin{equation}
 	\pi_{ME}=\binom{n}{\floor{\frac{n}{2}}}^{-1}\sum_{|A|=\floor{\frac{n}{2}}}\pi_A ,
 \end{equation}
where $\pi_A= \Tr \rho^2_A$.
This quantity is related to the (average) linear entropy $S_L=(1-\pi_{ME})2^{\floor{\frac{n}{2}}}/(2^{\floor{\frac{n}{2}}}-1)$ introduced in \cite{PhysRevA.69.052330}. Notice that the minimizer of the potential of multipartite entanglement maximizes the average linear entropy.

First, we note an  interesting phenomenon  that some constructions of $4$, $7$, and $8$-qubit PEME states are exactly MMESs \cite{Zha_2011,Zha_2012,Zha_2018}.
Take the $8$-qubit PEME state $\ket{T_4}$ in $(\mathbb{C}^2)^{\otimes 8}$  constructed in \autoref{sec:graphstate} as an example. The $8$-qubit state has $70$ four-party reductions $ \rho_A$, $56$ of which are maximally mixed. The value of $\pi_A$ for maximally mixed $\rho_A$ is $1/16$.
	Continuing the notation from \autoref{sec:graphstate}, there are two types of $A_{K\times \bar{K}}$ with $\card{K}=4$ that are not of full rank, six $K$'s with $i=2$ and a singular $A_{K\times \bar{K}}$, whose value of $\pi_{A}$ is $1/4$, and eight $K$'s with $i=0$ and a singular $A_{K\times \bar{K}}$, whose value of $\pi_{A}$ is $1/8$.
%		 that are not of full rank with $i = 1$. It can be calculated that the value of $\pi_{A}$ corresponding to $A_{4\times \bar{4}}$ with $i = 0$ is $1/4$, while the value of $\pi_{A}$ corresponding to $A_{4\times \bar{4}}$ with $i = 1$ is $1/8$.
%	{\color{red}(explain a bit more)}.
%There are two types of $\rho_A$ that are not maximally mixed notated as $\rho_{A1}$ and $\rho_{A2}$, of which there are $6$ $\rho_{A1}$ and $8$ $\rho_{A2}$  The value of $\pi_{A}$ for $\rho_{A1}$ is $1/4$, and for $\rho_{A2}$ it is $1/8$.
 The $\pi_{ME}$ of this PEME state is equal to $3/35$, which is equal to the smallest value of $\pi_{ME}$ obtained in \cite{Zha_2018}.

Second, some known MMESs give a good lower bound for $\Qex(n)$. For example, the $4$, $7$, and $8$-qubit MMESs constructed in \cite{Zha_2011,Zha_2012,Zha_2018} reach the $\Qex(n)$, thus they are in fact PEME states. In \cite{Zha_2020}, the construction of a $9$-qubit MMES gives a lower bound $\Qex(9)\geq 110$. The best known bounds of $\Qex(n)$ for $4\leq n\leq 12$ are summarized in \autoref{tab1}.

Third, we mention that there exists an MMES which is not a PEME state. For example,
%there are two classes of $4$-parties MESS \cite{Zha_2011}, only one of which is a $2$-EME state, while the other is not.
%%{\color{red}(state the counter example here)}.
%We have
\begin{equation}
	\ket{\phi}=\frac{1}{2}(\ket{0000}+\ket{0111}+\ket{1001}+\ket{1110})
\end{equation}
and
\begin{equation}
	\ket{M_4}=\frac{1}{\sqrt{6}}(\ket{0011}+\ket{1100}+\omega(\ket{1010}+\ket{0101})+\omega^2(\ket{1001}+\ket{0110})),
\end{equation}
where $\omega=e^{2\pi i/3}$. It can be checked that both $\ket{\phi}$ and $\ket{M_4}$ are MMESs. However, $\ket{\phi}$ is a $2$-EME state but $\ket{M_4}$ is not \cite{Zha_2011}.
Conversely, a PEME quantum state may also not be an MMES. For example, one can check that the $3$-uniform state $\ket{\psi}$ obtained by the orthogonal array $M_8$ in \cite{PhysRevA.99.042332} is a PEME state with eight qubits.
However, it can be calculated that $\ket{\psi}$ have a value of $\pi_{ME}=56 \times 1/16 + 14 \times 1/4 = 1/10 > 3/35$, which means that $\ket{\psi}$ is not an MMES. This also implies that $\ket{\psi}$ and $\ket{T_4}$ are not LU-equivalent.

However, it is interesting to note that from the literature of known MMESs \cite{Zha_2011,Zha_2012,Zha_2018,Zha_2020},
at least one of the states achieving the lower bound listed in \autoref{tab1}  is an MMES.
%if the example of lower bound listed in is unique, this example is an MMES, and if the example is not unique, at least one of the examples is an MMES.

%$\pi_A$ of $1/4$ for all its reductions to $4$-party that are not maximally mixed, which means that it is not possible for it to have a value of $\pi_{ME}$ of $3/35$.}
%\begin{table}
%	\renewcommand{\arraystretch}{1.3}
%	\setlength{\tabcolsep}{6.9mm}
%	\centering
%	\begin{tabular}{c|c|c|c|c|c|c}
%		n & 4 & 7 & 8 & 9 & 10 & 12 \\
%		\hline
%		upper bound & 4 \footnote{Higuchi and Sudbery \cite{HIGUCHI2000213}.} & 32
%		\footnote{Huber, G\"uhne and Siewert \cite{PhysRevLett.118.200502}.}
%		& 56 \footnote{From the upper bound of $\ex_{2m}(4m,K_{2m+1}^{2m})$ (\autoref{nonevenmod4} and \autoref{eqq}).} & 120 \footnote{From the upper bound of $\ex_{k}(2k+1,H_{k})$ (\autoref{thmodd} and \autoref{eq2k1}).} & 240 \footnote{From the upper bound of $\ex_{k}(2k,K_{k+2}^{k})$ (\autoref{nonevenmod42} and \autoref{eq2k}).}  & 792 \footnotemark[3] \\
%		\hline
%		lower bound & 4 \footnotemark[1] & 32 \footnote{Zha \emph{et al.} \cite{Zha_2012} and Goyeneche and Życzkowski \cite{RN407}.}  & 56 \footnote{Zha \emph{et al.} \cite{Zha_2018} and Li and Wang \cite{PhysRevA.99.042332}.} & 112 \footnote{From constructions by known graph states \cite{Lars}.} & 200 \footnotemark[8]
%		& 540 \footnote{From the (1,3,6)-circulant graph state \cite{DANIELSEN20061351}.}\\
%	\end{tabular}
%	\caption{Values of $\Qex(n)$ for small $n$.
%		%		{\color{red}(change some [10] to your thms)}
%	}
%	\label{tab1}
%\end{table}
\begin{table}
	\renewcommand{\arraystretch}{1.3}
	\setlength{\tabcolsep}{6.9mm}
	\centering
	\begin{tabular}{c|c|c|c|c|c|c|c}
	n & 4 & 7 & 8 & 9 & 10 & 11 & 12 \\
	\hline
	upper bound & 4 \footnote{Higuchi and Sudbery \cite{HIGUCHI2000213}.} & 32
	\footnote{Huber, G\"uhne and Siewert \cite{PhysRevLett.118.200502}.}
	  & 56 \footnote{From the upper bound of $\ex_{2m}(4m,K_{2m+1}^{2m})$ (\autoref{nonevenmod4} and \autoref{eqq}).} & 120 \footnote{From the upper bound of $\ex_{k}(2k+1,H_{k})$ (\autoref{thmodd} and \autoref{eq2k1}).} & 240 \footnote{From the upper bound of $\ex_{k}(2k,K_{k+2}^{k})$ (\autoref{nonevenmod42} and \autoref{eq2k}).} & 461 \footnote {AME states of eleven qubits do not exist \cite{PhysRevA.69.052330}.} & 792 \footnotemark[3] \\
	\hline
	lower bound & 4 \footnotemark[1] & 32 \footnote{Zha \emph{et al.} \cite{Zha_2012} and Goyeneche and Życzkowski \cite{RN407}.}  & 56 \footnote{Zha \emph{et al.} \cite{Zha_2018} and Li and Wang \cite{PhysRevA.99.042332}.} & 112 \footnote{From constructions by known graph states \cite{Lars}.} & 200 \footnotemark[9] & 396 \footnotemark[9]
	& 540 \footnote{From the (1,3,6)-circulant graph state \cite{DANIELSEN20061351}.}\\
\end{tabular}
	\caption{Values of $\Qex(n)$ for small $n$.
%		{\color{red}(change some [10] to your thms)}
		}
	\label{tab1}
\end{table}

%\begin{center}
%	\begin{tabular}{c|c|c|c|c|c|c}
%		n & 4 & 7 & 8 & 9 & 10 & 12 \\
%		\hline
%		upper bound & 4 \cite{HIGUCHI2000213} & 32 \cite{PhysRevLett.118.200502} & 56 & 120 \cite{PhysRevLett.118.200502} & 240 \cite{PhysRevLett.118.200502}  & 792\\
%		\hline
%		lower bound & 4 \cite{HIGUCHI2000213} & 32 \cite{Zha_2012} & 56 \cite{Zha_2018} & 110 \cite{Zha_2020} & 192  & 540 \cite{DANIELSEN20061351}\\
%	\end{tabular}
%\end{center}
\section{conclusion}\label{sec:con}

In summary, we give bounds on the maximum number of maximally mixed half-body reductions in an arbitrary qubit pure state, where the upper bound is given by the Tur\'an's number and the lower bound is given by explicitly constructing graph states. This allows us to show that  that $56$ is the largest number of maximally mixed $4$-party reductions in an $8$-qubit pure state.
For future work, the connected Tur\'an's problem and the construction of graph states are not fully resolved. Better solutions to these problems would lead to better results on the problem we focus. In particular, it is very interesting to determine the $\Qex(n)$ for a specific $n$.
%In particular, whether $\lim_{n \to \infty} \pi(n,f(n))$ exists also remains an open question.
 The estimation of the values of $\Qex(n,k)$ for $k < \floor{n/2}$ is also worth studying.

 	\section{Acknowledgments}
 	 The research of W. Zhang, Y. Ning and X. Zhang was supported by the
 	Innovation Program for Quantum Science and Technology 2021ZD0302902,
 	the NSFC under Grants No. 12171452 and No. 12231014, and the National
 	Key Research and Development Programs of China 2023YFA1010200 and
 	2020YFA0713100.
 	The research of F. Shi was supported by the HKU Seed Fund
 	for Basic Research for New Staff via Project 2201100596, Guangdong Natural
 	Science Fund via Project 2023A1515012185, National Natural Science Foun
 	dation of China (NSFC) via Project No. 12305030 and No. 12347104, Hong
 	Kong Research Grant Council (RGC) via No. 27300823, N\_HKU718/23,
 	and R6010-23, Guangdong Provincial Quantum Science Strategic Initiative
 	GDZX2200001.

\appendix
\section{A proof of \autoref{nonevenmod42}}\label{appendix1}
\begin{proof} Let $\ket{\psi}$ be an $n$-qubit pure state, where $n=2k$, $k \geq 2$ and $k\neq 3$.
	The proof is by contradiction. Suppose that there exists  $\mathcal{A} \subset [n]$ with
$\card{\mathcal{A}} = k+2$ such that the reduction of $\ket{\psi}$ to each $k$-party contained in $\mathcal{A}$ is maximally mixed.
	WLOG, assume $\mathcal{A} = [k+2]$.
% Suppose that $\mathcal{B} \subset \mathcal{A}$ with $\card{\mathcal{B}} = k$, then for contradiction we have that the reduction to $\mathcal{B}$, which we denote $\rho_\mathcal{B}$, is maximally mixed.
 Then for any $\mathcal{B} \subset \mathcal{A}$ with $\card{\mathcal{B}} = k$, by the fact that $\card{\bar{\mathcal{B}}}=k=\card{\mathcal{B}}$ and $\rho_\mathcal{B}$  share the same spectrum as $\rho_{\bar{\mathcal{B}}}$, we know that $\rho_{\bar{\mathcal{B}}}$ is maximally mixed as well.  Since $\bar{\mathcal{A}} \subset \bar{\mathcal{B}}$,
	we know that $\rho_{\bar{\mathcal{A}}}$ is maximally mixed.
	By \autoref{eqrhoA2}, we have
	\begin{equation}\label{4n+2eq01}
		\rho_{\mathcal{A}}^2={2^{2-k}}\rho_{\mathcal{A}}.
	\end{equation}
	
	Since every reduction of $\rho_{\mathcal{A}}$ to $k$-party is
	maximally mixed, we have
		\begin{equation}\label{4n+2eq03apa}
		\rho_{\mathcal{A}}=\frac{1}{2^{k+2}}( I_2+\sum_{j=1}^{k+2}P_{k+1}^{(\bar{j})}\otimes  I_2 ^{(j)}+P_{k+2}),
	\end{equation}where $(\bar{j}):=[k+2]\setminus\{j\}$, and $ I_2 ^{(j)}$ means an identity in the $j$th party. Similarly, for every  $(k+1)$-party $ \mathcal{C} \subset \mathcal{A} $,
	\begin{equation}\label{4n+2eq02}
		\rho_{\mathcal{C}}=\frac{1}{2^{k+1}}( I_2+P_{k+1}).
	\end{equation}

	%There are $k+2$ different terms $P_{k+1}^{[j]}$
	%, with $[j]$ indexing the $k+2$ different supports of weight $k+1$ terms within a $(k+2)$-party reduction, each having an identity on different positions.

	By \autoref{eqschdec},
	a Schmidt decomposition of the pure state $\ket{\psi}$ across the bipartition $\mathcal{A}\mid \bar{\mathcal{A}}$ yields

	\begin{equation}\label{4n+2eq04}
		\rho_{\mathcal{A}}\otimes I_2^{\otimes(k-2)}\ket{\psi}_{\mathcal{A}\bar{\mathcal{A}}}=2^{2-k}\ket{\psi}_{\mathcal{A}\bar{\mathcal{A}}},
	\end{equation}
   and across the bipartition $\mathcal{C}\mid \bar{\mathcal{C}}$ for any $ \mathcal{C} \subset \mathcal{A} $ with $\card{\mathcal{C}}=k+1$ yields
	\begin{equation}\label{4n+2eq05}
		\rho_{\mathcal{C}}\otimes I_2^{\otimes(k-1)}\ket{\psi}_{\mathcal{C}\bar{\mathcal{C}}}=2^{1-k}\ket{\psi}_{\mathcal{C}\bar{\mathcal{C}}}.
	\end{equation}

	Substituting   \autoref{4n+2eq03apa} into \autoref{4n+2eq04}, we have
	\begin{equation}\label{4n+2eq06}
		\frac{1}{2^{k+2}}( I_2+\sum_{j=1}^{k+2}P_{k+1}^{(\bar{j})}\otimes  I_2 ^{(j)}+P_{k+2})\otimes  I_2^{\otimes(k-2)}\ket{\psi}=2^{2-k}\ket{\psi}.
	\end{equation}
		Substituting   \autoref{4n+2eq02} into \autoref{4n+2eq05}, we have for any $ \mathcal{C} \subset \mathcal{A} $ with $\card{\mathcal{C}}=k+1$,
	\begin{equation}\label{4n+2eq07}
		\frac{1}{2^{k+1}}( I_2+P_{k+1})\otimes  I_2^{\otimes(k-1)}\ket{\psi}=2^{1-k}\ket{\psi}.
	\end{equation}
Note that $P_{k+1}$ in \autoref{4n+2eq07} is indeed $P_{k+1}^{(\bar{j})}$ for some $j\in [k+2]$. So 	\begin{equation}\label{4n+2eq081}
			 P_{k+1}^{(\bar{j})}\otimes  I_2^{\otimes (k-1)} \ket{\psi}= 3 \ket{\psi}
	\end{equation} for each $j\in [k+2]$.
Substituting \autoref{4n+2eq081} into \autoref{4n+2eq06}, we get
	\begin{equation}\label{4n+2eq08apa}
			P_{k+2}\otimes  I_2^{\otimes (k-2)} \ket{\psi}= 3(3-k) \ket{\psi}.
	\end{equation}
	
	Further, combining  \autoref{4n+2eq01} and \autoref{4n+2eq03apa}, we obtain
	
	\begin{equation}\label{4n+2eq03}
	( I_2+\sum_{j=1}^{k+2}P_{k+1}^{(\bar{j})}\otimes  I_2 ^{(j)}+P_{k+2})( I_2+\sum_{j=1}^{k+2}P_{k+1}^{(\bar{j})}\otimes  I_2 ^{(j)}+P_{k+2})=16( I_2+\sum_{j=1}^{k+2}P_{k+1}^{(\bar{j})}\otimes  I_2 ^{(j)}+P_{k+2}).
	\end{equation}
	We consider terms of odd weight on both sides of
\autoref{4n+2eq03}, and find a contradiction. Denote $M:=\sum_{j=1}^{k+2}P_{k+1}^{(\bar{j})}\otimes  I_2 ^{(j)}$. By \autoref{rule}, each term in the anticommutator $\{M,P_{k+2}\} $ from the left hand side has odd weight.  In the following, our discussion depends on whether $k$ is odd or even.

For odd $k$, collecting terms of odd weight on both sides of
\autoref{4n+2eq03} gives
	\begin{equation} \label{eqoddapa}
	2 P_{k+2} +	\{M,P_{k+2}\} = 16 P_{k+2}, \text{ that is, }\{M,P_{k+2}\} = 14 P_{k+2}.
	\end{equation}		
Doing a tensor of \autoref{eqoddapa} with the identity and multiplying $\ket{\psi}$ from the right leads to
	\begin{equation}\label{11apa}
		\{M,P_{k+2}\} \otimes  I_2^{\otimes (k-2)}\ket{\psi}= 14 P_{k+2} \otimes  I_2^{\otimes (k-2)}\ket{\psi}.
	\end{equation}
By \autoref{4n+2eq081} and \autoref{4n+2eq08apa}, we know that
\begin{equation}\label{eqmp}
  \left(\left(P_{k+1}^{(\bar{j})}\otimes  I_2 ^{(j)}\right)P_{k+2}\right)\otimes  I_2^{\otimes (k-2)}\ket{\psi}=\left(P_{k+2}\left(P_{k+1}^{(\bar{j})}\otimes  I_2 ^{(j)}\right)\right)\otimes  I_2^{\otimes (k-2)}\ket{\psi}=9(3-k)\ket{\psi}.
\end{equation}
Then \autoref{11apa} becomes
\[2(k+2)\times 9(3-k)\ket{\psi}=14\times 3(3-k)\ket{\psi},\] which is not possible except when $k=3$.

	%We insert \autoref{4n+2eq08apa} in \autoref{11apa} to obtain
%	\begin{equation}\label{12apa}
%		3(k+2) P_{k+2} \otimes  I_2^{\otimes (k-2)}\ket{\psi} = 7 P_{k+2} \otimes  I_2^{\otimes (k-2)}\ket{\psi}.
%	\end{equation}	
%	If $P_{k+2} \otimes  I_2^{\otimes (k-1)} \neq 0$ , we obtain a contradiction because
%	\begin{equation}
%		3\times (k+2) \neq 7
%	\end{equation}
%	Thus, $P_{k+2} \otimes  I_2^{\otimes (k-1)} \ket{\psi} = 0$.
%	But from the eigenvector relation in \autoref{4n+2eq08apa} this
%can only by possible if $k = 3$.

For even  $k$, collecting terms of odd weight on both sides of
\autoref{4n+2eq03} gives
\begin{equation} \label{eqevenapa}
	2 M +	\{M, P_{k+2}\} = 16 M, \text{ that is, }\{M, P_{k+2}\} = 14 M.
\end{equation}
Doing a tensor of \autoref{eqevenapa} with the identity and multiplying $\ket{\psi}$ from the right leads to
\begin{equation}\label{15apa}
	\{M, P_{k+2}\} \otimes  I_2^{\otimes (k-2)}\ket{\psi}= 14 M \otimes  I_2^{\otimes (k-2)}\ket{\psi}.
\end{equation}
By \autoref{4n+2eq08apa} and \autoref{eqmp}, we have
\[2(k+2)\times 9(3-k)\ket{\psi}=14\times 3(k+2)\ket{\psi},\]
which is impossible.

%\begin{equation}\label{16apa}
%	3(3-k) M \otimes  I_2^{\otimes (k-2)}\ket{\psi} = 7 M \otimes  I_2^{\otimes (k-2)}\ket{\psi}.
%\end{equation}	
%
%
%
%Since $P_{k+1}^{(\bar{j})}\otimes  I_2^{\otimes k-1} \ket{\psi}= 3 \ket{\psi} \neq 0$ , we obtain a contradiction
%\begin{equation}
%	3\times (3-k) \neq 7.
%\end{equation}

\end{proof}

\section{A proof of \autoref{thmodd}}\label{appendix2}
\begin{proof}
	Let $\ket{\psi}$ be an $n$-qubit pure state, where $n=2k+1$, $k \geq 3$ and $k\neq 5$. The proof is by contradiction. Suppose that there exists  $\mathcal{A} \subset [n]$ with
	$\card{\mathcal{A}} = k+2$ such that any reduction of $\ket{\psi}$ to  $k$-party $\mathcal{B}$ with $\card{\mathcal{B}\cap \mathcal{A}}=1$ or $k$ is maximally mixed.	WLOG, assume $\mathcal{A} = [k+2]$.
	%	Here we take $\mathcal{I}$ as the full set, then we have the complement of $\mathcal{A}$ as $\setnd{2n+4, 2n+5, \dots, 4n+2}$, denoted $A^\complement$.
%	For contradiction we have that any reductions to a $k$-subset $\mathcal{B}$ with $\card{\mathcal{B}\cap \mathcal{A}}=1$ or $k$, which we denote $\rho_\mathcal{B}$, is maximally mixed.
    Let $\bar{\mathcal{A}} = [2k+1]\setminus[k+2]$ and let $\mathcal{B}= \bar{\mathcal{A}} \cup \{1\}$.
%	Since $\bar{\mathcal{A}} = [2k+1]\setminus[k+2]$, there must exist a $\mathcal{B}$ that intersects $\mathcal{A}$ with $1$ such that $\bar{\mathcal{A}}$ is a subset of $\mathcal{B}$.
    By the fact that $\rho_{{\mathcal{B}}}$ is maximally mixed, we have that $\rho_{\bar{\mathcal{A}}}$ is maximally mixed.
	By \autoref{eqrhoA2} we have

	\begin{equation}\label{4n+2eq01apb}
		\rho_{\mathcal{A}}^2={2^{1-k}}\rho_{\mathcal{A}}.
	\end{equation}
	
	Since every reduction of $\rho_{\mathcal{A}}$ to $k$-party  is
maximally mixed, we have
		\begin{equation}\label{4n+2eq03apb}
		\rho_{\mathcal{A}}=\frac{1}{2^{k+2}}( I_2+\sum_{j=1}^{k+2}P_{k+1}^{(\bar{j})}\otimes  I_2 ^{(j)}+P_{k+2}),
	\end{equation}where $(\bar{j}):=[k+2]\setminus\{j\}$, and $ I_2 ^{(j)}$ means an identity in the $j$th party. Similarly, for every  $(k+1)$ party $ \mathcal{C} \subset \mathcal{A} $,
	
	\begin{equation}\label{4n+2eq02apb}
		\rho_{\mathcal{C}}=\frac{1}{2^{k+1}}( I_2+P_{k+1}).
	\end{equation}

%	There are $k+2$ different terms $P_{k+1}^{[j]}$, with $[j]$ indexing the $k+2$ different supports of weight $k+1$ terms within a $(k+2)$-party reduction, each having an identity on different positions.
	
		By \autoref{eqschdec},
	a Schmidt decomposition of the pure state $\ket{\psi}$ across the bipartition $\mathcal{A}\mid \bar{\mathcal{A}}$ yields
%	Operating in the same way as in \autoref{eqschdec},
%	a Schmidt decomposition of the pure state $\ket{\psi}$ across the bipartitions $\{1, 2, \ldots, k+2|k+3, \ldots, 2k+1\}$ yields
%	
	\begin{equation}\label{4n+2eq04apb}
		\rho_{\mathcal{A}}\otimes I_2^{\otimes(k-1)}\ket{\psi}_{\mathcal{A}\bar{\mathcal{A}}}=2^{1-k}\ket{\psi}_{\mathcal{A}\bar{\mathcal{A}}},
	\end{equation}
	  and across the bipartition $\mathcal{C}\mid \bar{\mathcal{C}}$ for any $ \mathcal{C} \subset \mathcal{A} $ with $\card{\mathcal{C}}=k+1$ yields
%    Similarly, for any $ \mathcal{C} \subset \mathcal{A} $ with $\card{\mathcal{C}}=k+1$, we have
	\begin{equation}\label{4n+2eq05apb}
		\rho_{\mathcal{C}}\otimes I_2^{\otimes(k)}\ket{\psi}_{\mathcal{C}\bar{\mathcal{C}}}=2^{-k}\ket{\psi}_{\mathcal{C}\bar{\mathcal{C}}}.
	\end{equation}

%	We decompose \autoref{4n+2eq04apb} and \autoref{4n+2eq05apb} in the Bloch representation. According to the fact that any reductions to $k$ parties from $\mathcal{A}$ or $\mathcal{C}$ are maximally mixed, we obtain the following equation:

		Substituting   \autoref{4n+2eq03apb} into \autoref{4n+2eq04apb}, we have
	\begin{equation}\label{4n+2eq06apb}
		\frac{1}{2^{k+2}}( I_2+\sum_{j=1}^{k+2}P_{k+1}^{(\bar{j})}\otimes  I_2 ^{(j)}+P_{k+2})\otimes  I_2^{(k-1)}\ket{\psi}=2^{1-k}\ket{\psi}.
	\end{equation}
		Substituting   \autoref{4n+2eq02apb} into \autoref{4n+2eq05apb}, we have for any $ \mathcal{C} \subset \mathcal{A} $ with $\card{\mathcal{C}}=k+1$,
	\begin{equation}\label{4n+2eq07apb}
		\frac{1}{2^{k+1}}( I_2+P_{k+1})\otimes  I_2^{k}\ket{\psi}=2^{-k}\ket{\psi}.
	\end{equation}
	Note that $P_{k+1}$ in \autoref{4n+2eq07apb} is indeed $P_{k+1}^{(\bar{j})}$ for some $j\in [k+2]$. So
	\begin{equation}\label{4n+2eq08apb1}
 P_{k+1}^{(\bar{j})}\otimes  I_2^{\otimes k} \ket{\psi}=  \ket{\psi},
	\end{equation}
	for each $j\in [k+2]$.
	Substituting \autoref{4n+2eq08apb1} into \autoref{4n+2eq06apb}, we get
		\begin{equation}\label{4n+2eq08apb2}
	P_{k+2}\otimes  I_2^{\otimes (k-1)} \ket{\psi}= (5-k) \ket{\psi}.
	\end{equation}
		Further, combining  \autoref{4n+2eq01apb} and \autoref{4n+2eq03apb}, we obtain
	\begin{equation}\label{4n+2eq09apb}
( I_2+\sum_{j=1}^{k+2}P_{k+1}^{[j]}\otimes  I_2 ^{(j)}+P_{k+2})( I_2+\sum_{j=1}^{k+2}P_{k+1}^{[j]}\otimes  I_2 ^{(j)}+P_{k+2})=8( I_2+\sum_{j=1}^{k+2}P_{k+1}^{[j]}\otimes  I_2 ^{(j)}+P_{k+2}).
	\end{equation}
	
	We consider terms of odd weight on both sides of
\autoref{4n+2eq09apb}, and find a contradiction. Denote $M:=\sum_{j=1}^{k+2}P_{k+1}^{(\bar{j})}\otimes  I_2 ^{(j)}$. By \autoref{rule}, each term in the anticommutator $\{M,P_{k+2}\} $ from the left hand side has odd weight.  In the following, our discussion depends on whether $k$ is odd or even.

	For odd $k$, collecting terms of odd weight on both sides of
	\autoref{4n+2eq09apb} gives
		\begin{equation}\label{eqodd1apb}
	2 P_{k+2} +	\{M,P_{k+2}\} = 8 P_{k+2}, \text{ that is, }\{M,P_{k+2}\} = 3 P_{k+2}.
	\end{equation}
	Doing a tensor of \autoref{eqodd1apb} with the identity and multiplying $\ket{\psi}$ from the right leads to
	\begin{equation}\label{11apb}
		\{M, P_{k+2}\} \otimes  I_2^{\otimes (k-1)}\ket{\psi}= 3 P_{k+2} \otimes  I_2^{\otimes (k-1)}\ket{\psi}.
	\end{equation}

	By \autoref{4n+2eq08apb1} and \autoref{4n+2eq08apb2}, we know that
	\begin{equation}\label{eqmp2}
		\left(\left(P_{k+1}^{(\bar{j})}\otimes  I_2 ^{(j)}\right)P_{k+2}\right)\otimes  I_2^{\otimes (k-1)}\ket{\psi}=\left(P_{k+2}\left(P_{k+1}^{(\bar{j})}\otimes  I_2 ^{(j)}\right)\right)\otimes  I_2^{\otimes (k-1)}\ket{\psi}=(5-k)\ket{\psi}.
	\end{equation}
  Then \autoref{11apb} becomes \[2(k+2)\times (5-k)\ket{\psi}=3\times (5-k)\ket{\psi},\] which is not possible except when $k=5$.

For even  $k$, collecting terms of odd weight on both sides of
\autoref{4n+2eq09apb} gives
\begin{equation} \label{eqeven1apb}
	2 M +	\{M, P_{k+2}\} = 8 M, \text{ that is, }\{M, P_{k+2}\} = 3 M.
\end{equation}
Doing a tensor of \autoref{eqeven1apb} with the identity and multiplying $\ket{\psi}$ from the right leads to
\begin{equation}\label{15apb}
	\{M, P_{k+2}\} \otimes  I_2^{\otimes (k-1)}\ket{\psi}= 3 M \otimes  I_2^{\otimes (k-1)}\ket{\psi}.
\end{equation}
By \autoref{4n+2eq08apb1} and \autoref{eqmp2}, we have
\[2(k+2)\times (5-k)\ket{\psi}=3\times (k+2)\ket{\psi},\]
which is not possible except when $k=2$.

\end{proof}

%{\color{red}(check the style of all references, make them consistent)}

\bibliographystyle{IEEEtran}
\bibliography{reference}
% \tableofcontents
\end{document}